\documentclass[12pt,reqno]{amsart}
\usepackage[left=1.25in,right=1.25in,top=1.25in,bottom=1.25in]{geometry}
\usepackage{graphicx}
\usepackage{color}
\usepackage[dvipsnames]{xcolor}
\usepackage{amsmath, amsthm, amssymb, amsfonts}
\usepackage{natbib}
\usepackage{har2nat}
\usepackage{tabularx}
\usepackage{caption}
\usepackage{parskip}
\usepackage[british]{babel}

\newtheorem{theorem}{Theorem}

\newtheorem{lemma}{Lemma}
\newtheorem{proposition}{Proposition}
\theoremstyle{definition}

\newcommand{\ul}{\underline}
\newcommand{\ol}{\overline}
\newcommand{\df}{\mathrm{d}}

\newcommand{\bdis}{\begin{displaymath}}
\newcommand{\edis}{\end{displaymath}}
\newcommand{\beq}{\begin{equation}}
\newcommand{\eeq}{\end{equation}}
\newcommand{\bea}{\begin{eqnarray*}}
\newcommand{\eea}{\end{eqnarray*}}
\newcommand{\bean}{\begin{eqnarray}}
\newcommand{\eean}{\end{eqnarray}}

\newcommand{\R}{\mathbb{R}}
\newcommand{\E}{\mathbb{E}}

\DeclareMathOperator*{\argmax}{arg\,max}

\usepackage[colorlinks=true,allcolors=Blue,linkcolor=Blue]{hyperref}

\begin{document}

\let\MakeUppercase\relax

\title{Persuasion Meets Delegation}
\author[\uppercase{Kolotilin and Zapechelnyuk}]{
\larger \textsc{Anton Kolotilin and Andriy
Zapechelnyuk}}
\date{\today}
\thanks{ \ \\
\textit{Kolotilin}: School of Economics, UNSW Business School,
Sydney, NSW 2052, Australia. {\it E-mail:} {akolotilin@gmail.com}. \\
\textit{Zapechelnyuk}: School of Economics and Finance, University of St Andrews, Castlecliffe, the Scores, St Andrews KY16 9AR, UK. {\it E-mail:} {az48@st-andrews.ac.uk.} \\ 
\ \\
We are grateful to Tymofiy Mylovanov, with whom we are working on related projects. We thank Ricardo Alonso, Kyle Bagwell, Benjamin Brooks, Deniz Dizdar, Piotr Dworczak, Alexander Frankel, Drew Fudenberg, Gabriele Gratton, Yingni Guo, Emir Kamenica, Navin Kartik, Ming Li, Hongyi Li, Carlos Oyarzun, Alessandro Pavan, Eric Rasmusen, Philip Reny, Joel Sobel, and Thomas Tr\"{o}ger for helpful comments and suggestions. We also thank participants at various seminars and conferences. Part of this research was carried out while Anton Kolotilin was visiting MIT Sloan School of Management, whose hospitality and support is greatly appreciated. Anton Kolotilin also gratefully acknowledges support from the Australian Research Council Discovery Early Career Research Award DE160100964.  Andriy Zapechelnyuk gratefully acknowledges support from the Economic and Social Research Council Grant ES/N01829X/1.
}

\begin{abstract} 
A principal can restrict an agent's information (the persuasion problem) or restrict an agent's discretion (the delegation problem). We show that these problems are generally equivalent --- solving one solves the other.
We use tools from the persuasion literature to generalize and extend many results in the delegation literature, as well as to address novel delegation problems, such as monopoly regulation with a participation constraint.

\bigskip

\noindent\emph{JEL\ Classification:}\ D82, D83, L43\newline


\

\end{abstract}

\maketitle

\newpage

\section{Introduction}\label{Intro}
There are two ways to influence decision making: delegation and persuasion.
The delegation literature, initiated by \citeasnoun{Holmstrom}, studies the design of decision rules. The persuasion literature, set in motion by \citeasnoun{KG}, studies the design of information disclosure rules. 

The delegation problem has been used to design organizational decision processes \cite{Dessein}, monopoly regulation policies \cite{AM}, and international trade agreements \cite{AB}. The persuasion problem has been used to design school grading policies \cite{OS}, internet advertising strategies \cite{RS}, and forensic tests \cite{KG}.

This paper shows that, under general assumptions, the delegation and persuasion problems are equivalent, thereby bridging the two strands of literature. The implication is that the existing insights and results in one problem can be used to understand and solve the other problem.

Both delegation and persuasion problems have a principal and an agent whose payoffs depend on an agent's decision and a state of the world. The sets of decisions and states are intervals of the real line. The agent's payoff function satisfies standard single-peakedness and sorting conditions. In a delegation problem, the agent privately knows the state and the principal commits to a set of decisions from which the agent chooses. In a persuasion problem, the principal designs the agent's information structure and the agent freely chooses a decision.
The principal's tradeoff is that giving more discretion to the agent in the delegation problem and disclosing more information to the agent in the persuasion problem allows for better use of information about the state, but limits control over the biased agent's decision.

We consider {\it balanced delegation} and {\it monotone persuasion} problems. In the balanced delegation problem, the principal may not be able to exclude certain indispensable decisions of the agent. This problem nests the standard delegation problem and includes, in particular, a novel delegation problem with an agent's participation constraint.

In the monotone persuasion problem, the principal chooses a monotone partitional information structure that either reveals the state or pools it with adjacent states. This problem incorporates constraints faced by information designers in practice. For example, a non-monotone grading policy that gives better grades to  worse performing students will be perceived as unfair and will be manipulated by strategic students. Moreover, in many special cases, optimal information structures are monotone partitions.

The main result of the paper is that the balanced delegation and monotone persuasion problems are strategically equivalent. For each primitive of one problem we explicitly construct an equivalent primitive of the other problem. This construction equates the marginal payoffs and swaps the roles of decisions and states in the two problems. Intuitively, decisions in the delegation problem play the role of states in the persuasion problem because the principal controls decisions in the delegation problem and (information about) states in the persuasion problem. It is worth noting that this equivalence result is fundamentally different from the revelation principle. Specifically, the sets of implementable (and, therefore, optimal) decision outcomes generally differ in the delegation and persuasion problems with the same payoff functions.

To prove the equivalence result, we show that the balanced delegation and monotone persuasion problems are equivalent to the following {\it discriminatory disclosure} problem. The principal's and agent's payoffs depend on an agent's binary action, a state of the world, and an agent's private type. The sets of states and types are intervals of the real line. The agent's payoff function is single-crossing in the state and type. The principal designs a menu of cutoff tests, where a cutoff test discloses whether the state is below or above a cutoff. The agent selects a test from the menu and chooses between {\it inaction} and {\it action} depending on his private type and the information revealed by the test.

To see why the discriminatory disclosure problem is equivalent to the balanced delegation problem, observe that the agent's essential decision is the selection of a cutoff test from the menu. Because the agent's payoff function is single-crossing in the state, the agent optimally chooses inaction/action if the selected test discloses that the state is below/above the cutoff. Thus, this problem can be interpreted as a delegation problem in which a delegation set is identified with a menu of cutoffs, and the agent's decision with his selection of a cutoff from the menu.

To see why the discriminatory disclosure problem is equivalent to the monotone persuasion problem, observe that each menu of cutoff tests defines a monotone partition of the state space. 
Because the agent's payoff function is single-crossing in the state, the agent's optimal choice between inaction and action is the same whether he observes the partition element that contains the state or the result of the optimally selected cutoff test.
Moreover, because the agent's payoff function is single-crossing in his type, the agent optimally chooses inaction/action if his type is below/above a threshold. Thus, this problem can be interpreted as a persuasion problem in which a monotone partition is identified with a menu of cutoffs, and the agent's decision with a threshold type.

We use our equivalence result to solve a monopoly regulation problem in which a welfare-maximizing regulator (principal) restricts the set of prices available to a monopolist (agent) who privately knows his cost. 
This problem was first studied by \citeasnoun{BM82} as a mechanism design problem with transfers. \citeasnoun{AM} pointed out that transfers between the regulator and monopolist are often forbidden, and thus, the monopoly regulation problem can be formulated as a delegation problem. \citeasnoun{AM} omitted the monopolist's participation constraint, so under their optimal regulation policy, the monopolist sometimes operates at a loss. \citeasnoun{AB2} characterized the optimal regulation policy taking the participation constraint into account.

The monopoly regulation problem, with and without the participation constraint, can be formulated as a balanced delegation problem. 
We provide an elegant method of solving this problem, by recasting it as a monotone persuasion problem and using a single result from the persuasion literature. When the demand function is linear and the cost distribution is unimodal, the optimal regulation policy takes a simple form that is often used in practice. The regulator imposes a price cap and allows the monopolist to choose any price not exceeding the cap. The optimal price cap is higher when the participation constraint is present; so the monopolist is given more discretion when he has an option to exit.

The literature has focused on {\it linear} delegation and {\it linear} persuasion in which the marginal payoffs are linear in the decision and the state, respectively. We show the equivalence of linear balanced delegation and linear monotone persuasion. We translate a linear delegation problem to the equivalent linear persuasion problem and solve it using methods in \citeasnoun{Kolotilin2017} and \citeasnoun{DM}. Specifically, we provide conditions under which a candidate delegation set is optimal. For an interval delegation set, these conditions coincide with those in \citeasnoun{AM}, \citeasnoun{AB}, and \citeasnoun{ABF}, but we impose weaker regularity assumptions. For a two-interval delegation set, our conditions are novel and imply special cases in \citeasnoun{MS1991} and \citeasnoun{AM}.

Our equivalence result can also be used to translate the existing results in nonlinear delegation problems to equivalent nonlinear persuasion problems, and vice versa. Nonlinear delegation is considered in \citeasnoun{Holmstrom}, \citeasnoun{AM}, \citeasnoun{AB}, and \citeasnoun{ABF}. Nonlinear persuasion is considered in \citeasnoun{RS}, \citeasnoun{KG}, \citeasnoun{Kolotilin2017}, \citeasnoun{DM}, and \citeasnoun{GS2018}.

The rest of the paper is organized as follows. In Section \ref{Model}, we define the balanced delegation and monotone persuasion problems. In Section \ref{s:equiv}, we present and prove the equivalence result. In Section \ref{s:applic}, we apply the equivalence result to solve a monopoly regulation problem. In Section \ref{s:LDLP}, we address the linear delegation problem using tools from the persuasion literature. In Section \ref{s:proof-ext}, we present the equivalence result under weaker assumptions. In Section \ref{s:conc}, we make concluding remarks. The appendix contains omitted proofs.

\section{Two Problems}\label{Model}

\subsection{Primitives}
There are a principal (she) and an agent (he). The agent's payoff $U(\theta,x)$ and principal's payoff $V(\theta,x)$ depend on a decision $x\in[0,1]$ and a state $\theta\in[0,1]$. The state is uniformly distributed. We assume that the marginal payoffs are continuous, and the agent's payoff is supermodular and concave in the decision,\footnote{We relax these assumptions in Section~\ref{s:proof-ext}.}

(A$_1$) $\frac{\partial}{\partial x} U(\theta,x)$ and $\frac{\partial}{\partial x} V(\theta,x)$ are well defined and continuous in $\theta$ and $x$;

(A$_2$) $\frac{\partial}{\partial x} U(\theta,x)$ is strictly increasing in $\theta$ and strictly decreasing in $x$.

A pair $(U,V)$ is called a {\it primitive} of the problem. Let $\mathcal P$ be the set of all primitives that satisfy assumptions (A$_1$) and (A$_2$).

We now describe two problems. 
In a delegation problem, the agent is fully informed and the principal restricts the agent's discretion. In a persuasion problem, the agent has full discretion and the principal restricts the agent's information.

In both problems, the principal chooses a closed subset $\Pi$ of $[0,1]$ that contains the elements $0$ and $1$. Let 
\[
{\bf \Pi}=\{\Pi\subset [0,1]: \text{$\Pi$ is closed and $\{0,1\}\subset \Pi$}\}.
\]
In the delegation problem, $\Pi$ describes a set of decisions from which the agent chooses. In the persuasion problem, $\Pi$ describes a partition of states that the agent observes.

\subsection{Balanced Delegation Problem.}\label{s:deleg}

Consider a primitive $(U_D,V_D)\in\mathcal P$, where we use subscript $D$ to refer to the delegation problem.  
The principal chooses a {\it delegation set} $\Pi\in \bf \Pi$. 
The agent privately observes the state $\theta$ and chooses a decision $x\in\Pi$ that maximizes his payoff,
\addtocounter{equation}{1}
\begin{equation}\label{E:x}
	x^*_D(\theta,\Pi)=\argmax\limits_{x\in \Pi}  U_D(\theta,x).\tag{\arabic{equation}a}
\end{equation}
The principal's objective is to maximize her expected payoff, 
\begin{equation*}\label{E:D}
\max\limits_{\Pi\in\bf\Pi} \E\big[V_D(\theta,x^*_D(\theta,\Pi))\big].\tag{\arabic{equation}b}
\end{equation*}
By (A$_2$), $x^*_D(\theta,\Pi)$ is single-valued for almost all $\theta$, so $\E\big[V_D(\theta,x^*_D(\theta,\Pi))\big]$ is well~defined.

The balanced delegation problem requires delegation sets to include the extreme decisions. 
On the one hand, this requirement can be made non-binding by defining the agent's and principal's payoffs on a sufficiently large interval of decisions so that the extreme decisions are never chosen (Appendix \ref{s:StdDel}). On the other hand, this requirement allows to include indispensable decisions of the agent, such as a participation decision (Sections \ref{s:applic} and \ref{s:linpart}).

\subsection{Monotone Persuasion Problem.}\label{S:monp}
Consider a primitive $(U_P,V_P)\in\mathcal P$, where we use subscript $P$ to refer to the persuasion problem.
The principal chooses a monotone partitional information structure that partitions the state space into convex sets: separating elements and pooling intervals.
A monotone partition is described by a set $\Pi\in {\bf \Pi}$ of boundary points of these partition elements. Let
\[
\ul\pi(\theta)=\sup \{\theta'\in\Pi:\theta'\le \theta\} \quad\text{and}\quad \ol\pi(\theta)=\inf \{\theta'\in\Pi:\theta'>\theta\}
\]
for $\theta\in [0,1)$, and $\ul \pi (1) = \ol \pi (1) =1$.
The partition element $\mu_\Pi(\theta)$ that contains $\theta\in[0,1]$ is given by
\[
\mu_\Pi(\theta)=\begin{cases}
\{\theta\},& \text{if $\ul\pi(\theta)=\ol\pi(\theta)$},\\
[\ul\pi(\theta),\ol\pi(\theta)),& \text{if $\ul\pi(\theta)<\ol\pi(\theta)$}.
\end{cases}
\]
For example, $\Pi=\{0,1\}$ is the uninformative partition that pools all states,\footnote{Formally, state $\theta=1$ is separated, but this is immaterial, because this event has zero probability.} and $\Pi=[0,1]$ is the fully informative partition that separates all states. 

The agent observes the partition element $\mu_\Pi(\theta)$ that contains the state $\theta$ and chooses a decision $x\in [0,1]$ that maximizes his expected payoff given the posterior belief about $\theta$,
\addtocounter{equation}{1}
\begin{equation}\label{E:y}
x^*_P(\theta,\Pi)\in \argmax\limits_{x\in [0,1]}\E\big [U_P(\theta',x) \big| \theta'\in \mu_\Pi(\theta)\big ].\tag{\arabic{equation}a}
\end{equation}
The principal's objective is to maximize her expected payoff, 
\begin{equation*}\label{E:P}
\max\limits_{\Pi\in\bf\Pi} \E\big[V_P(\theta,x^*_P(\theta,\Pi))\big].\tag{\arabic{equation}b}
\end{equation*}
By (A$_2$), $x^*_P(\theta,\Pi)$ is single-valued for all $\theta$, so $\E\big[V_P(\theta,x^*_P(\theta,\Pi))\big]$ is well defined.

The monotone persuasion problem requires information structures to be monotone partitions. On the one hand, this requirement is without loss of generality in many special cases, where optimal information structures are monotone partitions (Section~\ref{s:LDLP}). On the other hand, this requirement may reflect incentive and legal constraints faced by information designers. Monotone partitional information structures are widespread and include, for example, school grades, tiered certification, credit and consumer ratings.

\subsection{Persuasion versus Delegation}\label{S:Example} We now show that implementable outcomes differ in the persuasion and delegation problems with the same primitive $(U,V)\in \mathcal P$.

In the persuasion problem with $U(\theta,x)=-(\theta-x)^2$, consider a monotone partition $\Pi'$ that reveals whether the state is above or below $1/3$. Since $\theta$ is uniformly distributed on $[0,1]$, the induced decision of the agent is
\begin{equation*}\label{E:ex1}
x^*_P(\theta,\Pi')=\begin{cases}
\frac 1 6, &\text{if $\theta\in\left(0,\frac 1 3\right)$},\\
\frac 2 3,& \text{if $\theta \in \left ( \frac 1 3,1\right)$},
\end{cases}
\end{equation*}
where 1/6 is the midpoint between 0 and 1/3, and 2/3 is the midpoint between 1/3 and 1.

This outcome cannot be implemented in the delegation problem with the same primitive $U(\theta,x)=-(\theta-x)^2$. To see this, consider a delegation set $\Pi''$ that permits only two decisions, $1/6$ and $2/3$. The induced decision of the agent is
\begin{equation*}\label{E:ex2}
x^*_D(\theta,\Pi'')=\begin{cases}
\frac 1 6, & \text{if $\theta\in\left(0,\frac 5 {12}\right)$},\\
\frac 2 3,& \text{if $\theta\in \left( \frac 5 {12},1\right)$},
\end{cases}
\end{equation*}
where $ 5/12$ is the midpoint between $1/6$ and $2 /3$. Thus, the induced decisions in the persuasion and delegation problems differ on the interval of intermediate states between $1/3$ and $5/12$.\footnote{Since the outcome $x^*_P(\cdot,\Pi')$ cannot be implemented in the delegation problem, it cannot be implemented in the balanced delegation problem.}

Note that the outcome $x^*_P(\cdot,\Pi')$ is the first best for the principal whose payoff $V(\theta,x)$ is maximized at $x=1/6$ for states below $1/3$ and is maximized at $x=2/3$ for states above $1/3$. This first best is not implementable in the delegation problem with the same primitive. Conversely, the outcome $x^*_D(\cdot,\Pi'')$ is the first best for the principal whose payoff $V(\theta,x)$ is maximized at $x=1/6$ for states below $5/12$ and is maximized at $x=2/3$ for states above $5/12$. This first best is not implementable in the persuasion problem with the same primitive.\footnote{Similarly, if the principal's payoff $V(\theta,x)$ is maximized at decision $0$ for states below $1/2$ and is maximized at decision 1 for states above $1/2$, then the first best is implementable by the balanced delegation set $\{0,1\}$, but is not implementable by any monotone partition.}

\section{Equivalence}\label{s:equiv}

\subsection{Main Result\label{s:mr}} We use \citename{vnm}'s \citeyear{vnm} notion of strategic equivalence.
Primitives $(U_D, V_D)\in\mathcal P$ and $(U_P,V_P)\in\mathcal P$ of the balanced delegation and monotone persuasion problems are {\it equivalent} if there exist $\alpha>0$ and $\beta\in\R$ such that
\begin{equation*}
\E\big[V_D(\theta,x^*_D(\theta,\Pi))\big]=\alpha \E\big[V_P(\theta,x^*_P(\theta,\Pi))\big]+\beta \quad \text{for all $\Pi\in\bf\Pi$}.
\end{equation*}
That is, if $(U_D,V_D)$ and $(U_P, V_P)$ are equivalent, then, in both problems, the principal gets the same expected payoff, up to an affine transformation, for each $\Pi$; consequently, the principal's optimal solution is also the same.

\addtocounter{equation}{1}
\begin{theorem}\label{T:1}
For each $(U_D,V_D)\in \mathcal P$,  an equivalent  $(U_P, V_P)\in \mathcal P$ is given by
\begin{align}\label{E:DtoP}
U_P(\theta,x)= -\int_0^x  \left.\frac{\partial U_D(t,s)}{\partial s}\right|_{s=\theta}\df t \quad\text{and}\quad V_P(\theta,x)= -\int_0^x  \left.\frac{\partial V_D(t,s)}{\partial s}\right|_{s=\theta}\df t.\tag{\arabic{equation}a}
\end{align}
Conversely, for each $(U_P,V_P)\in \mathcal P$,  an equivalent  $(U_D, V_D) \in \mathcal P$ is given by
\begin{align}\label{E:PtoD}
U_D(\theta,x)= \int_x^1  \left.\frac{\partial U_P(s,t)}{\partial t}\right|_{t=\theta}\df s \quad\text{and}\quad V_D(\theta,x)= \int_x^1  \left.\frac{\partial V_P(s,t)}{\partial t}\right|_{t=\theta}\df s.\tag{\arabic{equation}b}
\end{align}
\end{theorem}

Theorem \ref{T:1} explicitly connects the two problems. In fact, we prove a stronger result than Theorem \ref{T:1}.  We show that if primitives $(U_D, V_D)\in\mathcal P$ and $(U_P,V_P)\in\mathcal P$ satisfy
\[%
\left.\frac{\partial U_D(\theta_D,x)} {\partial x}\right|_{x=\theta_P}=-\left.\frac{\partial U_P(\theta_P,x) }{\partial x}\right|_{x=\theta_D} \ \ \text{and} \ \
\left.\frac{\partial V_D(\theta_D,x)} {\partial x}\right|_{x=\theta_P}=-\left.\frac{\partial V_P(\theta_P,x) }{\partial x}\right|_{x=\theta_D}
\]
for all $\theta_D,\theta_P\in[0,1]$, then, in both problems, not only the principal but also the agent gets the same ex-ante expected payoff, up to a constant, for each $\Pi$.
This connection equates the marginal payoffs from a decision and swaps the roles of decisions and states in the two problems. In particular, the agent's payoff is supermodular in one problem whenever it is concave in the decision in an equivalent problem.

The rest of this section proves the equivalence of the balanced delegation and monotone persuasion problems, by showing that they are equivalent to a discriminatory disclosure problem. In this problem, the agent is privately informed and has a binary action. The principal designs a menu of cutoff tests, where a cutoff test discloses whether the state is below or above a cutoff. The agent selects a test from the menu and then chooses an action. 

\subsection{Discriminatory Disclosure Problem}
The agent chooses between actions $a=0$ and $a=1$. The agent's payoff $u(s,t)$ and principal's payoff $v(s,t)$ from $a=1$ depend on a state $s\in [0,1]$ and an agent's private type $t\in [0,1]$; the payoffs from $a=0$ are normalized to zero. The state and type are independently and uniformly distributed. We assume that:

(A$'_1$) $u(s,t)$ and $v(s,t)$ are continuous in $s$ and $t$;

(A$'_2$) $u(s,t)$ is strictly increasing in $s$ and strictly decreasing in $t$.

A pair $(u,v)$ that satisfies assumptions (A$'_1$) and (A$'_2$) is a primitive of this problem. 

The principal  designs a menu $\Pi\in \bf \Pi$ of cutoff tests. Each cutoff test $y\in \Pi$ discloses whether the state $s$ is at least $y$. 
The agent knows his private type $t$, selects a cutoff test $y$ from the menu $\Pi$, observes the result of the selected test, and then chooses between $a=0$ and $a=1$. 

\subsection{Equivalence to Balanced Delegation.}
Consider a discriminatory disclosure problem with a primitive $(u,v)$. A menu of cutoffs $\Pi$ can be interpreted as a delegation set, and the agent's selection of a cutoff $y\in \Pi$ as an agent's decision.   
Indeed, by (A$'_2$), the agent gets a higher payoff from $a=1$ when the state $s$ is higher; so either he optimally chooses $a=1$ whenever $s\ge y$, or makes a choice irrespective of the test result. But ignoring the test result is the same as selecting an uninformative test $y\in \{0,1\}\subset \bf \Pi$, and then choosing $a=1$ whenever $s\ge y$. Therefore, without loss of generality, after observing the result of the selected test, $s\geq y$ or $s<y$, the agent chooses $a=1$ if $s\ge y$ and $a=0$ if $s<y$.

Thus, the agent selects a test $y\in\Pi$ that maximizes his expected payoff,
\addtocounter{equation}{1}
\beq\label{E:PPXX}
y^*(t,\Pi)=\argmax_{y\in \Pi} \int_{y}^{1}u(s,t) \df s.\tag{\arabic{equation}a}
\eeq
The principal's objective is to maximize her expected payoff,
\begin{equation}\label{E:PPI}
\max\limits_{\Pi\in\bf\Pi} \E\left[\int_{y^*(t,\Pi)}^{1}v(s,t) \df s\right].\tag{\arabic{equation}b}
\end{equation}
By (A$'_2$), $y^*(t,\Pi)$ is single-valued for almost all $t$, so $\E\big[\int_{y^*(t,\Pi)}^{1}v(s,t) \df s\big]$ is well~defined.

The discriminatory disclosure problem \eqref{E:PPXX}--\eqref{E:PPI} is a balanced delegation problem \eqref{E:x}--\eqref{E:D} with 
\[
U_D(\theta,x)= \int_{x}^{1}\left.u(s,t) \right|_{t=\theta}\df s \quad \text{and} \quad V_D(\theta,x)= \int_{x}^{1}\left.v(s,t) \right|_{t=\theta}\df s.
\]
Conversely, for $(U_D,V_D)\in\mathcal P$, define  $(\bar U_D,\bar V_D)\in\mathcal P$ by
\[
\bar U_D(\theta,x)=U_D(\theta,x)-U_D(\theta,1) \quad \text{and}\quad \bar V_D(\theta,x)=V_D(\theta,x)-V_D(\theta,1).
\]
Note that, for each $\Pi$, the principal gets the same expected payoff in the balanced delegation problem with $(U_D,V_D)$ and $(\bar U_D,\bar V_D)$, up to a constant, $\E [V_D(\theta,1)]$.

The balanced delegation problem \eqref{E:x}--\eqref{E:D} with primitive $(\bar U_D,\bar V_D)$ is a discriminatory disclosure problem \eqref{E:PPXX}--\eqref{E:PPI} with 
\[
u(s,t)=-\frac{\partial U_D (t,s)}{\partial s} \quad \text{and} \quad v(s,t)=-\frac{\partial V_D (t,s)}{\partial s}.
\]
Finally, $(U_D,V_D)$ satisfies (A$_1$)--(A$_2$) if and only if $(u,v)$ satisfies (A$'_1$)--(A$'_2$).

\subsection{Equivalence to Monotone Persuasion.}
Consider a discriminatory disclosure problem with a primitive $(u,v)$. A menu $\Pi\in\bf\Pi$ defines a monotone partition of $[0,1]$. The agent's optimal choice between $a=0$ and $a=1$ is the same whether he observes the partition element $\mu_\Pi(s)$ or the result of the optimally selected cutoff test $y\in \Pi$.   Indeed, by (A$'_2$), the agent gets a higher payoff from $a=1$ when a partition element is higher; so he optimally chooses $a=1$ whenever the partition element is at least $y\in \Pi$. Therefore, the agent behaves as if he observes the partition element $\mu_\Pi(s)$ that contains the state $s$. 

Furthermore, by (A$_2'$), the agent gets a higher expected payoff from $a=1$ when his type $t$ is lower; so he optimally chooses $a=1$ whenever $t\le z$ for some $z\in [0,1]$. Therefore, without loss of generality, after observing $\mu_\Pi(s)$, the agent chooses a threshold type $z\in [0,1]$, and then $a=1$ if $t\leq z$ and $a=0$ if $t>z$.  

Thus, the agent chooses a threshold type $z\in [0,1]$ that maximizes his expected payoff
\addtocounter{equation}{1}
\begin{equation}\label{E:PPYY}
z^*(s,\Pi)=\argmax_{z\in[0,1]}\E\left.\left[\int_0^z u(s',t)\df t\right|s'\in\mu_\Pi(s)\right].\tag{\arabic{equation}a}
\end{equation}
The principal's objective is to maximize her expected payoff,
\begin{equation}\label{E:PPM}
\max\limits_{\Pi\in\bf\Pi} \E\left[\int_0^{z^*(s,\Pi)} v(s,t)\df t \right].\tag{\arabic{equation}b}
\end{equation}
By (A$'_2$), $z^*(s,\Pi)$ is single-valued for all $s$, so $\E\big[\int_0^{z^*(s,\Pi)} v(s,t)\df t \big]$ is well defined.

The discriminatory disclosure problem \eqref{E:PPYY}--\eqref{E:PPM} is a monotone persuasion problem \eqref{E:y}--\eqref{E:P} with
\[
U_P(\theta,x)=\int_0^x \left.u(s,t)\right|_{s=\theta}\df t \quad \text{and} \quad V_P(\theta,x)=\int_0^x \left.v(s,t)\right|_{s=\theta}\df t.
\]
Conversely, for $(U_P,V_P)\in\mathcal P$, define $(\bar U_P,\bar V_P)\in\mathcal P$ by
\[
\bar U_P(\theta,x)=U_P(\theta,x)-U_P(\theta,0) \quad \text{and}\quad \bar V_P(\theta,x)=V_P(\theta,x)-V_P(\theta,0).
\]
Note that, for each $\Pi$, the principal gets the same expected payoff in the monotone persuasion problem with $(U_P,V_P)$ and $(\bar U_P,\bar V_P)$, up to a constant, $\E [V_P(\theta,0)]$.

The monotone persuasion problem \eqref{E:y}--\eqref{E:P} with primitive $(\bar U_P,\bar V_P)$ is a discriminatory disclosure problem \eqref{E:PPYY}--\eqref{E:PPM} with 
\[
u(s,t)=\frac{\partial U_P(s,t)}{\partial t} \quad \text{and} \quad v(s,t)=\frac{\partial V_P(s,t)}{\partial t}.
\]
Finally, $(U_P,V_P)$ satisfies (A$_1$)--(A$_2$) if and only if $(u,v)$ satisfies (A$'_1$)--(A$'_2$).

\section{Application to Monopoly Regulation}\label{s:applic}

We consider the classical problem of monopoly regulation as in \citeasnoun{BM82}. The monopolist privately knows his cost and chooses a price to maximize profit. The welfare-maximizing regulator can restrict the set of prices the monopolist can choose from, for example, by imposing a price cap. Following \citeasnoun{AM}, we assume that the demand function is linear and the marginal cost has a unimodal distribution. Importantly, unlike in \citeasnoun{BM82}, transfers between the monopolist and regulator are prohibited.

We study two versions of this problem: (i) with the monopolist's participation constraint, as in \citeasnoun{BM82} and \citeasnoun{AB2}, and (ii) without any participation constraint, as in \citeasnoun{AM}. We formulate both versions as balanced delegation problems. We then find the equivalent monotone persuasion problems and solve them using a single result from the persuasion literature. We show that, in both versions, the welfare-maximizing regulator imposes a price cap, which is higher when the participation constraint is present.

\subsection{Setup}

The demand function is $q=1-x$ where $x$ is the price and $q$ is the quantity demanded at this price. The cost of producing quantity $q$ is $\gamma q$. The marginal cost $\gamma\in \left[0,1\right]$ has a distribution $F$ that admits a strictly positive, continuous, and unimodal density~$f$. 

The monopolist's (agent's) payoff is the profit  
\[\label{e:monU}
U_D(\gamma,x)=(x-\gamma)(1-x).
\]
The regulator's (principal's) payoff is the sum of the profit and consumer surplus,
\[\label{e:monV}
V_D(\gamma,x)=U_D(\gamma,x)+\tfrac{1}{2}(1-x)^2.
\]
The regulator chooses a set of prices $\Pi\subset [0,1]$ available to the monopolist. The monopolist privately observes the marginal cost $\gamma$ and chooses a price $x$ from $\Pi$ to maximize profit $U_D(\gamma,x)$.

We first assume that the monopolist cannot be forced to operate at a loss. Formally, the monopolist can always choose to produce zero quantity, which is the same as setting price $x=1$; so $1\in \Pi$. 

Notice that selling at zero price gives a lower profit than not producing at all, regardless of the value of the marginal cost. Thus, allowing the price $x=0$ does not affect the monopolist's behavior; so, without loss of generality, $0\in \Pi$. To sum up, the regulator chooses a closed set of prices $\Pi\subset [0,1]$ that contains $0$ and $1$; so $\Pi \in \bf \Pi$.

To interpret this problem as a balanced delegation problem defined in Section \ref{s:deleg}, we change the variable $\theta=F(\gamma)$, so that $\theta$ is uniformly distributed on $[0,1]$. The monopolist's and regulator's payoffs are now given by
\beq\label{Eq:DelegUV}
U_D(\theta,x)=(x-F^{-1}(\theta))(1-x) \quad \text{and} \quad V_D(\theta,x)=U_D(\theta,x)+\tfrac{1}{2}(1-x)^2.
\eeq

\subsection{Analysis}\label{s:translation}
By Theorem~\ref{T:1}, an equivalent primitive $(U_P,V_P)$ of the monotone persuasion problem is given by\footnote{In Appendix \ref{s:int}, we provide an interpretation of this persuasion problem.}
\begin{align*}
U_P(\theta,x)&=-\int_0^x (1+F^{-1}(t)-2\theta)\df t=\int_0^{F^{-1}(x)} (2\theta-1-\gamma)f(\gamma)\df\gamma,\\
V_P(\theta,x)&=-\int_0^x (F^{-1}(t)-\theta)\df t=\int_0^{F^{-1}(x)} (\theta-\gamma)f(\gamma)\df\gamma,
\end{align*}
where the set of decisions is $[0,1]$, the set of states is $[0,1]$, and the state is uniformly distributed. In this problem, the payoffs are linear in the state; so only the posterior mean state matters. Under this assumption on the payoffs, optimal information structures are characterized in the literature (Section~\ref{LPSD}). We now illustrate how tools from this literature can address the monopoly regulation problem.

Let $m_\Pi(\theta)=\E[\theta'|\theta'\in \mu_\Pi(\theta)]$ be the posterior mean state induced by a partition element $\mu_\Pi(\theta)$ of a monotone partition $\Pi$. 
Since $U_P$ is linear in $\theta$, the agent's optimal decision depends on $\mu_\Pi(\theta)$ only through $m_\Pi(\theta)$; so $x^*_P(\theta,\Pi)=\bar x^*(m_\Pi(\theta))$ with
\[
\bar x^*(m)= \argmax_{x\in [0,1]}\int_0^{F^{-1}(x)} (2m-1-\gamma)f(\gamma)\df\gamma
=F(2m-1),
\]
where, by convention, $F(2m-1)=0$ if $2m-1\leq 0$ and $F(2m-1)=1$ if $2m-1\geq 1$.

Since $V_P$ is linear in $\theta$, the principal's expected payoff given $\mu_\Pi (\theta)$ is a function $\nu$ that depends on $\mu_\Pi(\theta)$ only through $m_\Pi(\theta)$:
\begin{align}
\nu(m)&=\int_0^{F^{-1}(\bar x^*(m))}(m-\gamma) f(\gamma)\df\gamma=\int_0^{2m-1}(m-\gamma) f(\gamma)\df\gamma.
\label{Eq:VM}
\end{align}
Thus, the principal's objective is to maximize the expectation of $\nu$,
\[\label{E:obj}
\max\limits_{\Pi\in\bf\Pi} \E\big[\nu(m_\Pi(\theta))\big].
\]

\begin{figure}
\begin{tabular}[b]{cc}
\hspace*{-0.15in}\includegraphics[width=230pt]{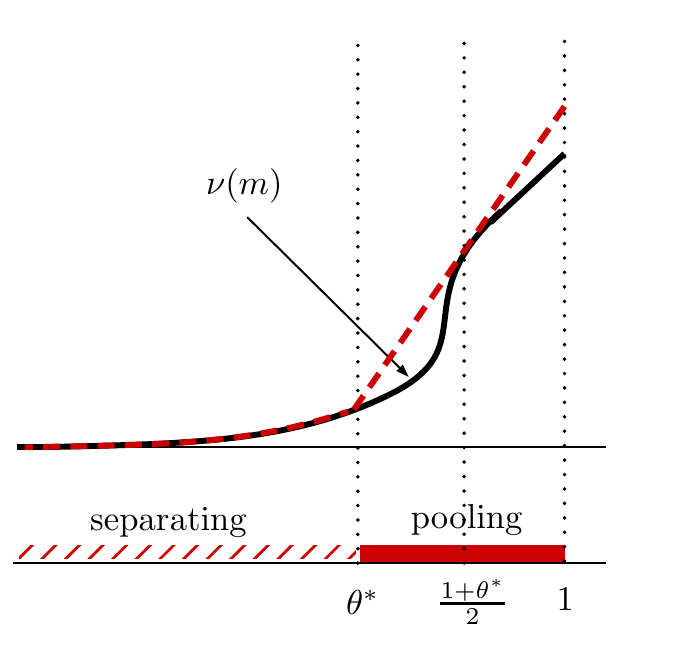} & \hspace*{-0.3in}\includegraphics[width=230pt]{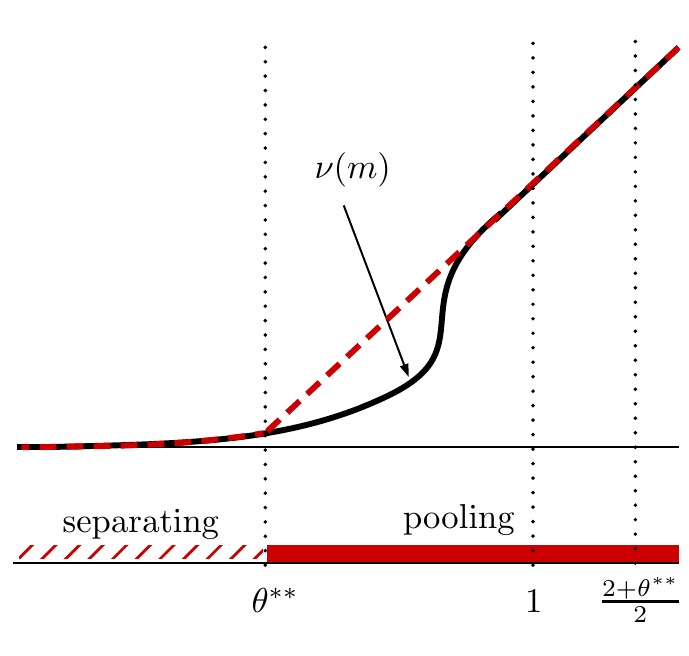}\\
\footnotesize{(a)  With Participation Constraint}&\footnotesize{(b) Without Participation Constraint}
\end{tabular}
\caption{Optimal Monotone Partitions}
\label{F:A}
\end{figure}

The curvature of $\nu$ determines the form of the optimal monotone partition. Because the density $f$ is unimodal,  $\nu$ is $S$-shaped (see Figure~\ref{F:A}(a)). Thus, the optimal monotone partition is an {\it upper-censorship}: the states below a cutoff $\theta^*$ are separated, and the states above $\theta^*$ are pooled and induce the posterior mean state equal to $(1+\theta^*)/2$.
\begin{proposition}\label{P:st}
Let $\gamma_m\in(0,1)$ be the mode of the density $f$. The set $\Pi^*=[0,\theta^*]\cup\{1\}$ is optimal, where $\theta^*\in \left(\gamma_m,(1+\gamma_m)/2\right)$ is the unique solution to
\beq\label{Eq:UC}
\nu\left(\frac{1+\theta^*}{2}\right)- \nu\left(\theta^*\right)=\left(\frac{1+\theta^*}{2}-\theta^*\right)\nu'\left(\frac{1+\theta^*}{2}\right).
\eeq
\end{proposition}
Since upper censorship $\Pi^*=[0,\theta^*]\cup\{1\}$ is optimal in the monotone persuasion problem, the same delegation set is optimal in the monopoly regulation problem. That is, the regulator imposes the price cap $\theta^*$, thus implementing the price function
\[
x^*(\gamma)=\begin{cases}
	\frac{1+\gamma} 2, & \text{if $\gamma<2\theta^*-1$},\\
	\theta^*,& \text{if $2\gamma\in (\theta^*-1,\theta^*)$}, \\
	1, & \text{if $\gamma>\theta^*$}.
\end{cases}
\]
In words, the monopolist chooses not to participate if his marginal cost $\gamma$ is above the price cap $\theta^*$. The participating monopolist chooses his preferred price $(1+\gamma)/2$ if it is below the cap, and he chooses the cap otherwise.

\subsection{Analysis without Participation Constraint}\label{s:wo}
We now assume that the regulator can force the monopolist to operate even when making a loss. That is, the regulator can choose any set of prices, without an additional constraint to include the price $x=1$.

To interpret this problem as a balanced delegation problem defined in Section \ref{s:deleg}, we observe that when the price $x$ is sufficiently high or sufficiently low, both the monopolist and regulator prefer intermediate prices. Thus, the requirement to include sufficiently extreme prices into the delegation set is not binding.

Specifically, consider $U_D$ and $V_D$ given by (\ref{Eq:DelegUV}) and defined on the domain of prices $[0,2]$. The regulator chooses a closed delegation set $\Pi\subset [0,2]$. Observe that, regardless of the marginal cost $\gamma\in[0,1]$, the regulator's payoff is negative if $x>1$ and zero if $x=1$. Therefore, an optimal delegation set $\Pi^*$ must contain a price $x_0\in[0,1]$. Moreover, the monopolist prefers $x_0$ to $0$ and to any price $x\geq 2$. Therefore, $\Pi^*\cup \{0,2\}$ implements the same price function as $\Pi^*$.

We thus obtain a balanced delegation problem, up to rescaling of the monopolist's decision. Using Theorem~\ref{T:1}, we find the equivalent primitive of the monotone persuasion problem. %

For comparability, it is convenient to rescale the state in the monotone persuasion problem, so that it is uniformly distributed on $[0,2]$. Analogously to Section \ref{s:translation}, the principal's objective now is to choose a monotone partition $\Pi\subset [0,2]$ such that $\{0,2\}\in \Pi$ to maximize the expectation of $\nu$ given by \eqref{Eq:VM} defined on $[0,2]$. Because $\nu$ is still $S$-shaped  (see Figure \ref{F:A}(b)), the optimal monotone partition is an upper-censorship: the states below a cutoff $\theta^{**}$ are separated, and the states above $\theta^{**}$ are pooled and induce the posterior mean state equal to $(2+\theta^{**})/2$.
\addtocounter{proposition}{-1}
\renewcommand{\theproposition}{\arabic{proposition}$'$}
\begin{proposition} \label{P:st1}
Let $\gamma_m\in(0,1)$ be the mode of the density $f$. The set $\Pi^{**}=[0,\theta^{**}]\cup\{2\}$ is optimal, where $\theta^{**}\in \left(0,(1+\gamma_m)/2\right)$ is the unique solution to
\beq\label{Eq:UC1}
\nu\left(\frac{2+\theta^{**}}{2}\right)- \nu\left(\theta^{**}\right)=\left(\frac{2+\theta^{**}}{2}-\theta^{**}\right)\nu'\left(\frac{2+\theta^{**}}{2}\right).
\eeq
\end{proposition}
\renewcommand{\theproposition}{\arabic{proposition}}
Proposition~\ref{P:st1} implies that $[0,\theta^{**}]$ is the optimal delegation set in the monopoly regulation problem without the participation constraint. That is, the regulator imposes the price cap $\theta^{**}$, thus implementing the price function
\[
x^{**}(\gamma)=\begin{cases}
	\frac{1+\gamma} 2, & \text{if $\gamma<2\theta^{**}-1$},\\
	\theta^{**},& \text{if $\gamma > 2\theta^{**}-1$}. 
\end{cases}
\]
In words, the monopolist chooses his preferred price $(1+\gamma)/2$ if it is below the price cap, and he chooses the cap otherwise.

\subsection{Discussion}
The optimal regulation policy takes the form of a price cap, regardless of whether the monopolist's participation constraint is present. However, the optimal price cap is higher when the participation constraint is present, as follows from Propositions~\ref{P:st} and \ref{P:st1}. Indeed, since $\nu$ is concave on $[(1+\gamma_m)/2,2]$ and
\[
\frac{1+\gamma_m}{2}<\frac{1+\theta^*}{2}<1<\frac{2+\theta^{**}}{2}<2,
\]
the slope of $\nu$ is higher at $(1+\theta^*)/2$ than at $(2+\theta^{**})/2$; so (\ref{Eq:UC}) and (\ref{Eq:UC1}) imply that $\theta^*>\theta^{**}$ (see Figures \ref{F:A}(a) and \ref{F:A}(b)).

We now build the intuition for why the optimal price cap is higher when the participation constraint is present. The first-order condition (\ref{Eq:UC}) for the optimal price cap $\theta^*$ can be written as
\beq\label{E:15p}
\int _{2\theta^*-1}^{\theta^*} (\theta^*-\gamma)f(\gamma)\df \gamma =\frac{1}{2}(1-\theta^*)^2f(\theta^*), \tag{\ref{Eq:UC}$'$}
\eeq
where the left-hand side and right-hand side correspond to the regulator's marginal gain and marginal loss of decreasing the price cap by $\df \theta$. The gain is that the monopolist with the cost $\gamma\in (2\theta^*-1,\theta^*)$ now chooses the decreased price cap $\theta^*-\df \theta$, which is closer to his cost $\gamma$. The loss is that the monopolist with the cost $\gamma\in (\theta^*-\df\theta,\theta^*)$ now chooses to exit. 

Instead, if the regulator does not take into account that the monopolist with the cost higher than the price cap exits, then the first-order condition (\ref{Eq:UC1}) for the price cap $\theta^{**}$ can be written as 
\beq\label{E:21p}
\int _{2\theta^{**}-1}^{\theta^{**}} (\theta^{**}-\gamma)f(\gamma)\df \gamma =\int _{\theta^{**}}^{1} (\gamma-\theta^{**} )f(\gamma)\df \gamma. \tag{\ref{Eq:UC1}$'$}
\eeq
The regulator's marginal gain here is the same. But the marginal loss is that the monopolist with the cost $\gamma\in (\theta^{**},1)$ chooses the decreased price cap $\theta^{**}-\df \theta$, which is further from his cost $\gamma$. 

Intuitively, the marginal loss in (\ref{E:15p}) is higher than in (\ref{E:21p}), because all surplus is lost if the monopolist exits, but only a part of surplus is lost if the monopolist sets the price further away from his cost.\footnote{The right-hand side of (\ref{E:15p}) can be expressed as $\int _{\theta^{*}}^{1} (\gamma-\theta^{*} )f(\theta^*)\df \gamma$.  This marginal loss is higher than in  (\ref{E:21p}), because $f(\theta^*)> f(\gamma)$ for $\gamma>\theta^*>\gamma_m$ by the unimodality of $f$.}
This suggests that the regulator should give more discretion to the monopolist when she is concerned that the monopolist can exit.

\section{Linear Delegation and Linear Persuasion}\label{s:LDLP}

\subsection{Setup\label{LPSD}}
Consider a primitive $(U_D,V_D)\in \mathcal P$ of the balanced delegation problem that satisfies 
\beq\label{E:LinD}
\frac{\partial U_D (\theta,x)}{\partial x} = b(\theta) - c(x) \quad \text{and} \quad \frac{\partial V_D (\theta,x)}{\partial x}  = d(\theta) -  c(x),
\eeq
where $b$, $c$, and $d$ are continuous, and $b$ and $c$ are strictly increasing. By Theorem \ref{T:1}, for each $(U_D,V_D)\in \mathcal P$ that satisfies \eqref{E:LinD}, an equivalent primitive $(U_P,V_P)\in \mathcal P$  of the monotone persuasion problem satisfies
\beq\label{E:LinP}
\frac{\partial U_P (\theta,x)}{\partial x} = c(\theta) - b(x) \quad \text{and} \quad \frac{\partial V_P (\theta,x)}{\partial x}=c(\theta)-d(x). 
\eeq
Conversely, for each $(U_P,V_P)\in \mathcal P$ that satisfies \eqref{E:LinP}, an equivalent primitive $(U_D,V_D)\in \mathcal P$  of the balanced delegation problem satisfies \eqref{E:LinD}.

We call $(U_D,V_D)$ and $(U_P,V_P)$ that satisfy \eqref{E:LinD} and  \eqref{E:LinP} {\it linear}, because the marginal payoffs from a decision are linear, respectively, in a transformation of the decision, $c(x)$, and in a transformation of the state, $c(\theta)$.

Linear delegation (albeit without the inclusion of the extreme decisions) has been studied by \citeasnoun{Holmstrom}, \citeasnoun{MS1991}, \citeasnoun{MS2006},  \citeasnoun{AM}, \citeasnoun{GHPS}, \citeasnoun{KovacMyl}, \citeasnoun{AB}, and \citeasnoun{ABF}.\footnote{\citeasnoun{AB2} study linear delegation with a participation constraint. Despite the same assumptions on the payoffs, linear delegation is conceptually different from veto-based delegation of \citeasnoun{KM2001}, \citeasnoun{Dessein}, and \citeasnoun{TM08} and from limited-commitment delegation of \citeasnoun{KLL}.}

Linear persuasion (albeit without the restriction to monotone partitions) has been studied by \citeasnoun{KG}, \citeasnoun{GK-RS}, \citeasnoun{KMZL}, \citeasnoun{Kolotilin2017}, and \citeasnoun{DM}.\footnote{\citeasnoun{DM} and \citeasnoun{KL} study linear monotone persuasion. \citeasnoun{DM} provide conditions under which monotone partitions are optimal among all information structures. \citeasnoun{KL} characterize optimal monotone partitions when they differ from optimal information structures.}

It is convenient to represent a linear monotone persuasion problem as 
\beq\label{E:obj1}
\max\limits_{\Pi\in\bf\Pi} \E\big[\nu(m_\Pi(\theta))\big],
\eeq
for some function $\nu$, where $m_\Pi(\theta)=\E[c(\theta')|\theta'\in \mu_\Pi(\theta)]$. We can derive $\nu$ from $(U_P,V_P)$ as follows. Since $U_P$ is linear in $c(\theta)$, the agent's optimal decision depends on $\mu_\Pi(\theta)$ only through $m_\Pi(\theta)$; so $x^*_P(\theta,\Pi)=\bar x^*(m_\Pi(\theta))$. Moreover, since $V_P$ is linear in $c(\theta)$, the principal's expected payoff is a function $\nu$ that depends on $\mu_\Pi(\theta)$ only through $m_\Pi(\theta)$,
\[
\nu(m)=V_P(m,\bar x^*(m)).
\] 
Conversely, for each function $\nu$, a monotone persuasion problem reduces to \eqref{E:obj1} if $(U_P,V_P)$ satisfies \eqref{E:LinP} with $b(x)=x$, $c(\theta)=\theta$, and $d(x)=-\nu'(x)$.

\subsection{Optimal Linear Delegation}\label{s:old}
We now generalize and extend the existing results in the literature on linear delegation, using the tools from the literature on linear persuasion (\citename{Kolotilin2017}, \citeyear*{Kolotilin2017}, and \citename{DM}, \citeyear*{DM}).

We consider a linear delegation problem where the agent's and principal's payoffs are given by \eqref{E:LinD}, the set of states is a compact interval, and the set of decisions is the real line.  
Without loss of generality, we rescale the state and decision so that $b(\theta)=\theta$, where the state $\theta\in [0,1]$ has a distribution $F$ that admits a strictly positive and continuous density~$f$.
For the problem to be well defined, we assume that there exist the agent's and principal's preferred decisions for each state.
Specifically, we assume that there exist $x',x''\in \R$ such that $c(x')\leq \theta\leq c(x'')$ and $c(x')\leq d(\theta)\leq c(x'')$ for all $\theta$. 

In this problem, the principal chooses a compact subset $\Pi\subset \R$ to maximize her expected payoff, 
\begin{equation*}
\max\limits_{\Pi\in{\bf\Pi}(\R)} \E\big[V_D(\theta,x^*_D(\theta,\Pi))\big],
\end{equation*}
where ${\bf\Pi}(\R)$ is the set of all compact subsets of $\R$. 
As we show in Appendix \ref{s:StdDel}, this problem can be formulated as a balanced delegation problem with a sufficiently large compact set of decisions $[\ul y,\ol y]$. Notice that the decision is rescaled so that the principal chooses $\Pi\in {\bf \Pi}([\ul y,\ol y])$ where 
\beq\label{E:Pi}
{\bf \Pi}([\ul y,\ol y])=\{\Pi\subset [\ul y, \ol y]: \text{$\Pi$ is closed and $\{\ul y,\ol y\}\subset \Pi$}\}.
\eeq
By Theorem \ref{T:1}, an equivalent primitive $(U_P,V_P)$ of the monotone persuasion problem is given~by
\begin{align*}
	U_P(\theta,x)&=\int_0^{F^{-1}(x)} (c(\theta)-t)f(t)\df t \ \ \text{and} \ \ V_P(\theta,x)=\int_0^{F^{-1}(x)} (c(\theta)-d(t))f(t)\df t,
\end{align*}
where the set of decisions is $[0,1]$, the set of states is $[\ul y,\ol y]$, and the state is uniformly distributed. Notice that the state is rescaled so that the principal chooses a monotone partition $\Pi\in {\bf \Pi}([\ul y,\ol y])$.

Since $U_P$ is linear in $c(\theta)$, the agent's optimal decision depends on $\mu_\Pi(\theta)$ only through 
\[m_\Pi(\theta)=\E[c(\theta')|\theta'\in \mu_\Pi(\theta)]=\begin{cases}
\{c(\theta)\},& \text{if $\ul\pi(\theta)=\ol\pi(\theta)$},\\
\frac{\int_{\ul\pi(\theta)}^{\ol\pi(\theta)} c(\theta')\df \theta'}{\ol\pi(\theta)-\ul\pi(\theta)},& \text{if $\ul\pi(\theta)<\ol\pi(\theta)$},
\end{cases}
\] 
where $\ul\pi(\theta)$ and $\ol\pi(\theta)$ are defined in Section~\ref{S:monp}; so $x^*_P(\theta,\Pi)=\bar x^*(m_\Pi(\theta))$ with
\[
\bar x^*(m) = \argmax_{x\in [0,1]}\int_0^{F^{-1}(x)} (m-t)f(t)\df t =F(m),
\]
where, by convention, $F(m)=0$ if $m\leq 0$ and $F(m)=1$ if $m\geq 1$.

Since $V_P$ is linear in $c(\theta)$, the principal's expected payoff given $\mu_\Pi (\theta)$ is a function $\nu$ that depends on $\mu_\Pi(\theta)$ only through $m_\Pi(\theta)$,
\begin{equation}
\begin{aligned}
\nu(m)&=\int_0^{F^{-1}(\bar x^*(m))}(m-d(t)) f(t)\df t = \int_0^{m}(m-d(t)) f(t)\df t.
\label{Eq:VM-6}
\end{aligned}
\end{equation}
The next theorem verifies whether a candidate delegation set is optimal.

\begin{theorem}\label{T:3}
$\Pi^*\in{\bf \Pi}([\ul y, \ol y])$ is optimal if $p_{\Pi^*}(m)$ is convex and $p_{\Pi^*}(m)\ge \nu(m)$ for all $m\in c([\ul y, \ol y])$, where, for all $s \in \mathbb [\ul y, \ol y]$,
\[
p_{\Pi^*}(c(s))=\nu(m_{\Pi^*}(s))+\nu'(m_{\Pi^*}(s))(c(s)-m_{\Pi^*}(s)).
\]
\end{theorem}
\begin{proof} Consider any $\Pi\in {\bf \Pi}([\ul y, \ol y])$. The theorem follows from
\begin{align*}
\E[\nu(m_{\Pi^*}(s))]-\E[\nu(m_{\Pi}(s))]	&= \E[p_{\Pi^*}(c(s))]-\E[\nu(m_{\Pi}(s))]\\
&\ge \E[p_{\Pi^*}(c(s))]-\E[p_{\Pi^*}(m_{\Pi}(s))]\\
&=\E [\E [p_{\Pi^*}(c(s'))]-p_{\Pi^*}(m_{\Pi}(s)) | s'\in \mu_\Pi (s)]]\ge 0,
\end{align*}
where the first equality holds by definition of $p_{\Pi^*}$, the first inequality holds by $p_{\Pi^*}\ge \nu$, the second equality holds by the law of iterated expectations, and the second inequality holds by Jensen's inequality applied to convex $p_{\Pi^*}$.
\end{proof}

Using Theorem \ref{T:3}, we find sufficient conditions under which one- or two-interval delegation sets are optimal (see Figure \ref{f:3}).

\begin{figure}
\begin{tabular}{ccc}
\hspace*{-0.6in}\includegraphics[width=170pt]{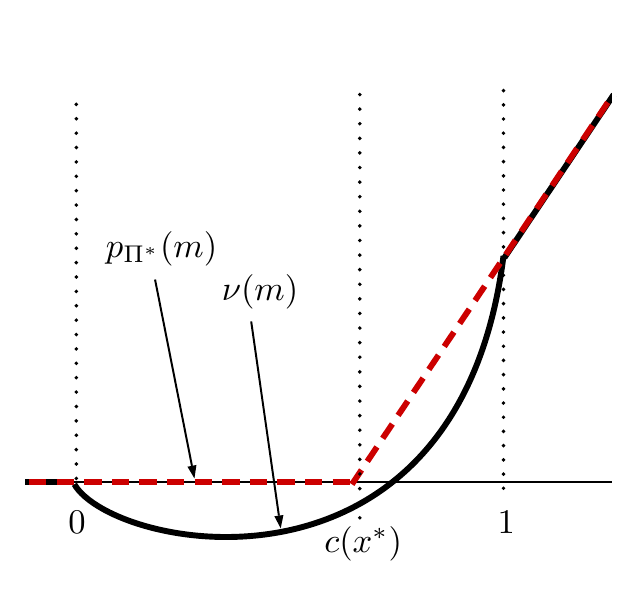} & \hspace*{-0.1in}\includegraphics[width=170pt]{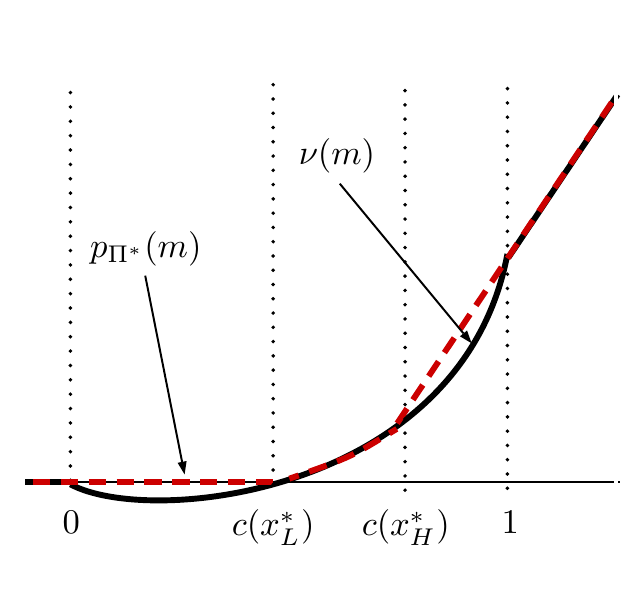} & \hspace*{-0.1in}\includegraphics[width=170pt]{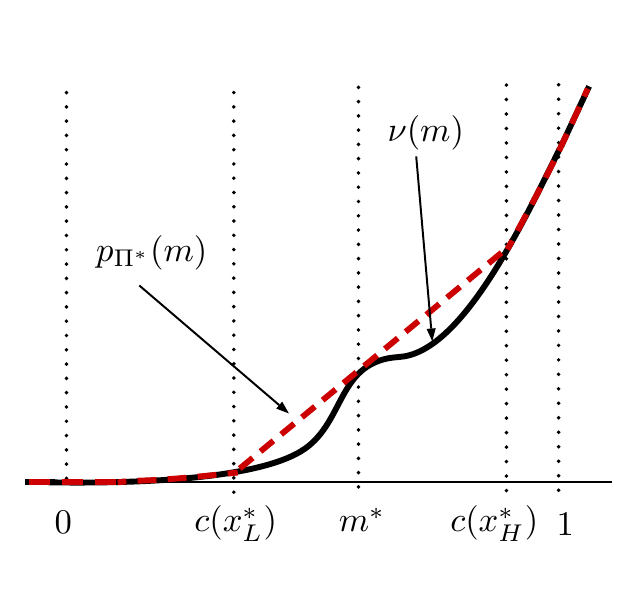}\\
\footnotesize{(a) Part 1}&\footnotesize{(b) Part 2} & \footnotesize{(c) Part 3}
\end{tabular}
\caption{Three Parts in Proposition \ref{P:D}}
\label{f:3}
\end{figure}

\begin{proposition}\label{P:D}\hfill

\nopagebreak
$1.$ A delegation set $\{x^*\}$ is optimal if
\begin{align*}
	&\nu(m) \leq 0 \ \text{ for all $m \leq c(x^*)$},\\
	&\nu(m) \leq \nu(1) +  m - 1 \ \text{ for all $m\geq c(x^*)$},\\
	&\nu(1) +c(x^*) - 1=0.
\end{align*}

$2.$ A delegation set $[x^*_L,x^*_H]$ with $x^*_L<x^*_H$ is optimal if
\begin{align*}
	&\nu(m) \ \text{is convex on $[c(x^*_L),c(x^*_H)]$},\\
	&\nu(m) \leq   0 \ \text{ for all $m \leq c(x^*_L)$ with equality at $m=c(x^*_L)$,}\\
	&\nu(m) \leq \nu(1) +  m - 1 \ \text{for all $m\geq c(x^*_H)$ with equality at $m=c(x^*_H)$},\\
	&\nu'(0 +) \ge 0\text{ if $c(x^*_L)=0$ and $\nu'(1 -) \leq 1$ if $c(x^*_H)=1$}.
\end{align*}

$3.$ A delegation set $(-\infty,x^*_L]\cup [x^*_H,\infty)$  with $x^*_L<x^*_H$ is optimal if\,\footnote{The set of the agent's preferred decisions is $[c^{-1}(0),c^{-1}(1)]$. Thus, delegation set $(-\infty,x^*_L]\cup [x^*_H,\infty)$ implements the same outcome as set $[y^*_L,x^*_L]\cup [x^*_H,y^*_H]$ does where $y^*_L=\min\{x^*_L,c^{-1}(0)\}$ and $y^*_H=\max\{x^*_H,c^{-1}(1)\}$.}
\begin{align*}
	&\nu(m) \ \text{is convex on $(-\infty,c(x^*_L)]$ and on $[c(x^*_H),\infty)$},\\
	&\nu(m) \leq   \nu(m^*)+\nu'(m^*) (m - m^*)  \text{ for all $m \in[c(x^*_L),c(x^*_H)]$}\\ 
	& \quad \quad \quad \  \text{with equality at $m=c(x^*_L)$ and $m=c(x^*_H)$},\\
	&\nu'(m^*) \ge 0\text{ if $c(x^*_L)=0$ and $\nu'(m^*) \le 1$ if $c(x^*_H)=1$}. 
\end{align*}
where
\[
m^*=\frac{1}{x^*_H-x^*_L}{\int_{x^*_L}^{x^*_H} c(s)\df s}.
\]
\end{proposition}

The literature on linear delegation has focused on characterizing sufficient conditions for interval delegation to be optimal. The conditions in \citename{AB} (\citeyear*{AB}, Proposition 1) coincide with those in Proposition \ref{P:D} (part 2), but we impose weaker regularity assumptions on the primitives. \citename{AB} (\citeyear*{AB}, Proposition 2) show that these conditions are also necessary.\footnote{\citename{AM} (\citeyear*{AM}, Propositions 3, 6, 7) also show the necessity and sufficiency of these conditions but only for quadratic payoffs ($c(x)=x$).} 

The sufficient conditions in \citename{ABF} (\citeyear*{ABF}, Proposition 1) for degenerate interval delegation to be optimal coincide with those in Proposition \ref{P:D} (part 1), but again we impose weaker regularity assumptions on the primitives. \citename{AM} (\citeyear*{AM}, Proposition 1) show that these conditions are also necessary.\footnote{They show the necessity only for quadratic payoffs ($c(x)=x$), but we can use their argument to show the necessity for non-quadratic payoffs. To this end, suppose first that there exists $m'<c(x^*)$ such that $\nu(m')>0$ and let $x'<x^*$ be such that $\int_{x'}^{x^*}c(s)\df s =  (x^*-x') m'$. The principal's expected payoff is larger by $(x^*-x') \nu(m')>0$ for delegation set $\{x',x^*\}$  than for $\{x^*\}$. The case in which there exists $m'>c(x^*)$ such that $\nu(m')>\nu(1)+m'-1$ is similar and omitted.}

The sufficient conditions in Proposition \ref{P:D} (part 3) for two-interval delegation to be optimal are novel. This result implies \citename{MS1991} (\citeyear*{MS1991}, Proposition 4) and \citename{AM} (\citeyear*{AM}, Result 6).

Using Proposition \ref{P:D}, we now characterize optimal delegation sets in prominent cases.

\begin{proposition}\label{P:Infl}\hfill

$1$ (convex). If $\nu(m)$ is convex on $[0,1]$, then there exist $x_L^*\le x^*_H$ such that delegation set $[x_L^*,x_H^*]$ is optimal.

$2$ (concave). If $\nu(m)$ is concave on $[0,1]$, then there exist $x_L^*\le x^*_H$ such that delegation set $\{x_L^*,x_H^*\}$ is optimal.

$3$ (convex-concave). If $\nu(m)$ is convex on $[0,\tilde m]$ and concave on $[\tilde m,1]$ for some $0<\tilde m<1$, then there exist $x_L^*\le x^*_M \le x^*_H$ such that delegation set $[x_L^*,x_M^*]\cup \{x^*_H\}$ is optimal.

$4$ (concave-convex). If $\nu(m)$ is concave on $[0,\tilde m]$ and convex on $[\tilde m,1]$ for some $0<\tilde m<1$, then there exist $x_L^*\le x^*_M \le x^*_H$ such that delegation set $\{x_L^*\}\cup [x^*_M,x^*_H]$ is optimal.
\end{proposition} 

\citename{ABF} (\citeyear*{ABF}, Proposition 2) coincides with Proposition \ref{P:Infl} (part 1), but again we impose weaker regularity assumptions on the primitives. \citeasnoun{MS2006} assume that $d(x)=x+\delta$ and $f(\theta)-\delta f'(\theta)\geq 0$, so that $\nu$ is convex on $[0,1]$. 

Parts 2--4 in Proposition \ref{P:Infl} are novel. \citeasnoun{MS1991} assume uniform density $f(\theta)=1$ and constant slope $d'(x)=k>0$, so that $\nu$ is convex on $[0,1]$ if $k<2$ and concave on $[0,1]$ if $k>2$. In the monopoly regulation problem of Section~\ref{s:applic}, $\nu$ is convex-concave on $[0,1]$ because $f$ is unimodal, but $\nu$ would be concave-convex on $[0,1]$ if $f$ was uniantimodal.

Finally, we comment on the limits of the tools from the persuasion literature for the delegation problem. The persuasion literature studies optimal information structures without the restriction to monotone partitions. 
Under the conditions in Theorem~\ref{T:3}, the set $\Pi$ is optimal among all information structures. These conditions are not necessary, because $\Pi$ may be optimal among monotone partitions, but not among all information structures. However, as shown by \citename{DM} (\citeyear*{DM}, Theorem~3), these conditions become necessary if $\nu$ is {\it affine-closed} (intuitively, if it has at most one interior peak). Even when optimal information structures are not monotone partitions, it may be possible to characterize optimal monotone partitions using tools from \citeasnoun{KL}.

\subsection{Optimal Linear Delegation with Participation Constraint}\label{s:linpart}
We now consider a linear delegation problem from Section~\ref{s:old} with the only difference that the set of decisions is $[x_0,\infty)$ and the decision $x_0$ must always be permitted by the principal. We assume that $c(x_0)\leq 0$ and that there exists $x''> x_0$ such that $\theta\leq c(x'')$ and $d(\theta)\leq c(x'')$ for all $\theta\in [0,1]$.  As follows from Appendix \ref{s:StdDel}, this problem can be formulated as a balanced delegation problem with a sufficiently large compact set of decisions $[x_0,\ol y]$. That is, the principal chooses $\Pi\in {\bf \Pi}([x_0,\ol y])$ to maximize the expectation of $\nu$ given by \eqref{Eq:VM-6}. Thus, Theorem \ref{T:3} continues to hold with $\ul y = x_0$.

As an illustration, we use Theorem \ref{T:3} to find sufficient conditions under which an optimal delegation set takes the form of a floor on the decisions.
\addtocounter{proposition}{-2}
\renewcommand{\theproposition}{\arabic{proposition}$'$}
\begin{proposition}\label{P:7} A delegation set $\{x_0\}\cup [x^*,\infty)$ with $x^*\geq x_0$ is optimal if
\begin{align*}
&\nu(m) \ \text{is convex on $[c(x^*),\infty)$},\\
&\nu(m) \leq   \nu(m^*)+\nu'(m^*) (m - m^*) \ \text{ for all $m \in [c(x_0), c(x^*)]$  with equality at $m=c(x^*)$},\\
&\nu'(0 +) \ge 0 \text{ if $c(x^*)=0$ and $\nu'(1 -) \leq 1$ if $c(x^*)=1$},
\end{align*}
where
\[
m^*=
\begin{cases}
c(x_0),& \text{if $x^*= x_0$},\\
\frac{1}{x^*-x_0} \int_{x_0}^{x^*} c(s)\df s, & \text{if $x^*> x_0$}.
\end{cases}
\]
\end{proposition}
\renewcommand{\theproposition}{\arabic{proposition}}

Proposition~\ref{P:7} can be used to derive conditions for a price cap (equivalently, a quantity floor) to be optimal in a monopoly regulation problem of Section~\ref{s:applic} with more general primitives specified as follows. The inverse demand function is $P(q)$. The marginal cost $\gamma\in [\ul \gamma, \ol \gamma ]$, with $0\leq \ul\gamma< \ol \gamma < \infty$, has a strictly positive and continuous density. The regulator's payoff is a weighted sum, with weights $\lambda\in (0,1)$ and $1-\lambda$, of the profit and consumer surplus,
\[
V(\gamma,q)=\lambda (P(q)-\gamma)q+ (1-\lambda) \left( \int_0^{q} P(s) \df s - P(q) q \right).
\]
We assume that the monopolist can always choose to produce zero quantity (exit).

The regulator's problem reduces to a balanced delegation problem in which the regulator chooses a set of quantities $\Pi\in {\bf \Pi} ([0,\ol q])$, where $P(\ol q) \leq \ul \gamma$, available to the monopolist. 
Up to rescaling of the state and decision, all assumptions imposed in this section are satisfied in each of the following three cases:

$1.$ $P(q)=A-B\ln q$ with $A\in \R$ and $B>0$;

$2.$ $P(q)=A-B q^C$ with $A>0$, $B>0$, and $C>0$;

$3.$ $P(q) = A - B q^{C}$ with $A<\ul \gamma$, $B<0$, and $C\in (-1,0)$.

For each of these cases, Proposition~\ref{P:7} gives sufficient conditions on the cost distribution, demand parameters, and payoff weights for price-cap regulation to be optimal. These conditions can be compared with those in \citeasnoun{AB2} who consider a similar setting but focus on delegation sets under which the monopolist never chooses to exit.

\section{General Equivalence}\label{s:proof-ext}

We now generalize Theorem \ref{T:1}. We assume that the marginal payoffs are integrable and single crossing, rather than continuous and monotone. We then discuss how this result can be used to deal with arbitrary distributions of the state, which may have atoms and zero density.

\subsection{Primitives}
We first define single crossing properties.
A function $\phi(y_i,y_{-i})$ satisfies

(i) {\it upcrossing in $y_i$} if, for each $y_{-i}$,
\[
\phi(y_i',y_{-i})\ge (>)\, 0  \implies  \phi(y_i'',y_{-i})\ge (>)\, 0 \ \ \text{whenever $y_i''>y_i'$};
\]

(ii) {\it aggregate upcrossing in $y_i$} if, for each distribution $H$ of $y_{-i}$, 
\[
\int \phi(y_i',y_{-i})\df H(y_{-i})\ge (>)\, 0  \implies \int \phi(y_i'',y_{-i})\df H(y_{-i})\ge(>)\, 0 \ \ \text{whenever $y_i''>y_i'$};
\]

(iii) {\it downcrossing (aggregate downcrossing) in $y_i$} if $-\phi(y_i,y_{-i})$ satisfies upcrossing (aggregate upcrossing) in $y_i$.\footnote{\citeasnoun{Quah2012} characterize conditions for aggregate single crossing. In particular, $\phi(y_i,y_{-i})$ satisfies aggregate single crossing in $y_i$ if $\phi(y_i,y_{-i})$ is monotone in $y_i$.}

The agent's payoff $U(\theta,x)$ and the principal's payoff $V(\theta,x)$ depend on a decision $x\in[0,1] $ and a state $\theta\in[0,1]$. The state is uniformly distributed. We assume that 

($\bar{\rm A}_1$) $U(\theta,x)$ and $V(\theta,x)$  are absolutely continuous in $x$; $\frac{\partial}{\partial x} U(\theta,x)$ and $\frac{\partial}{\partial x} V(\theta,x)$ are integrable in $\theta$. 

In addition, for the balanced delegation problem, we assume that

 ($\bar{\rm A}_2^D$)  $\frac{\partial}{\partial x} U(\theta,x)$ satisfies downcrossing in $x$ and aggregate upcrossing in $\theta$;
 
for the monotone persuasion problem, we assume that

 ($\bar{\rm A}_2^P$)  $\frac{\partial}{\partial x} U(\theta,x)$ satisfies upcrossing in $\theta$ and aggregate downcrossing in $x$.
 
The balanced delegation and monotone persuasion problems are defined as in Section \ref{Model}. But, unlike in Section \ref{Model}, the agent's optimal correspondences $x^*_D(\theta,\Pi)$ and $x^*_P(\theta,\Pi)$ may not be single-valued. Depending on which optimal decisions are selected by the agent, the principal can obtain different expected payoffs. Denote by
\begin{equation*}
\E\big[V_D(\theta,x^*_D(\theta,\Pi))\big] \quad\text{and}\quad \E\big[V_P(\theta,x^*_P(\theta,\Pi))\big]
\end{equation*}
the sets of the principal's expected payoffs resulting from all integrable selections from the correspondences $x^*_D(\theta,\Pi)$ and  $x^*_P(\theta,\Pi)$ \cite{Aumann}. 
 
\subsection{Equivalence}
Let $\bar{\mathcal P}_D$ be the set of all primitives $(U_D,V_D)$ that satisfy assumptions ($\bar{\rm A}_1$) and ($\bar{\rm A}_2^D$), and let $\bar{\mathcal P}_P$ be the set of all primitives $(U_P,V_P)$ that satisfy assumptions ($\bar{\rm A}_1$) and ($\bar{\rm A}_2^P$).

Primitives $(U_D,V_D)\in \bar{\mathcal P}_D$ and $(U_P,V_P)\in \bar{\mathcal P}_P$  are {\it equivalent} if there exist $\alpha>0$ and $\beta\in\R$ such that
\begin{equation*}
\E\big[V_D(\theta,x^*_D(\theta,\Pi))\big]=\alpha \E\big[V_P(\theta,x^*_P(\theta,\Pi))\big]+ \beta \quad \text{for all $\Pi\in{\bf \Pi}$}.
\end{equation*}
That is, if $(U_D,V_D)$ and $(U_P, V_P)$ are equivalent, then, in both problems, the principal gets the same {\it sets} of expected payoffs, up to an affine transformation, for each $\Pi$.

\addtocounter{theorem}{-2}
\renewcommand{\thetheorem}{\arabic{theorem}$'$}
\begin{theorem}\label{T:2}
For each $(U_D,V_D)\in\bar{\mathcal P}_D$,  an equivalent  $(U_P,V_P)\in\bar{\mathcal P}_P$ is given by \eqref{E:DtoP}. Conversely, for each $(U_P,V_P)\in\bar{\mathcal P}_P$,  an equivalent  $(U_D,V_D)\in\bar{\mathcal P}_D$ is given by \eqref{E:PtoD}.
\end{theorem}
\renewcommand{\thetheorem}{\arabic{theorem}}

The principal's set of expected payoffs is a singleton for each $\Pi$ (and thus the principal's maximization problem is well defined) if the optimal correspondence is single-valued for almost all $\theta$. This property holds in the delegation problem if $\frac{\partial}{\partial x}U_D(\theta,x)$ satisfies {\it strict} aggregate upcrossing in $\theta$. Similarly, this property holds in the persuasion problem if $\frac{\partial}{\partial x}U_P(\theta,x)$ satisfies {\it strict} aggregate downcrossing in $x$.

Alternatively, the principal's maximization problem can be defined by specifying a selection rule from the agent's optimal correspondence, such as the max rule (where the agent chooses the principal's most preferred decision) or the min rule (where the agent chooses the principal's least preferred decision).

\subsection{General Distributions}
Consider a balanced delegation or monotone persuasion problem with the state $\omega\in[\ul \omega,\ol\omega]$ that has a distribution $F$ (possibly, with atoms and zero density). To apply Theorem \ref{T:2}, we redefine the state to be $\theta$ uniformly distributed on $[0,1]$, and define $\omega=F^{-1}(\theta)$, where $F^{-1}(\theta)=\inf \{\omega\in [\ul \omega,\ol \omega]:\theta \leq F(\omega)\}$ is the quantile function.

With this change of variable, atoms in $F$ translate into intervals where $\frac{\partial}{\partial x}U(F^{-1}(\theta),x)$ and $\frac{\partial}{\partial x}V(F^{-1}(\theta),x)$ are constant in $\theta$; zero-density intervals in $F$ translate into points where $\frac{\partial}{\partial x}U_D(F^{-1}(\theta),x)$ and $\frac{\partial}{\partial x}V_D(F^{-1}(\theta),x)$ have simple discontinuities in $\theta$. This change of variable preserves both integrability in $\theta$ assumed in ($\bar{\rm A}_1$) and single-crossing in $\theta$ assumed in ($\bar{\rm A}_2^D$)/($\bar{\rm A}_2^P$).

For the case of atomless distributions with a strictly positive density, Theorem \ref{T:2} can be conveniently expressed as follows.
Let $U_D(\omega_D,x)$ and $V_D(\omega_D,x)$ satisfy assumptions ($\bar{\rm A}_1$) and ($\bar{\rm A}_2^D$), and let  $U_P(\omega_P,x)$ and $V_P(\omega_P,x)$ satisfy assumptions ($\bar{\rm A}_1$) and ($\bar{\rm A}_2^P$). Suppose that $\omega_D\in[0,1]$ and $\omega_P\in[0,1]$ are distributed with strictly positive densities $f_D$ and $f_P$. Applying Theorem \ref{T:2} to the primitives with redefined states, and changing states back to $\omega_D$ and $\omega_P$, we obtain that $(U_D,V_D)$ and $(U_P,V_P)$ are equivalent if, for all $\omega_D,\omega_P\in[0,1]$, 
\begin{align*}
&\left.\frac{\partial U_D(\omega_D,x)} {\partial x}\right|_{x=\omega_P}f_D(\omega_D)=-\left.\frac{\partial U_P(\omega_P,x) }{\partial x}\right|_{x=\omega_D}f_P(\omega_P),\\
&\left.\frac{\partial V_D(\omega_D,x)} {\partial x}\right|_{x=\omega_P}f_D(\omega_D)=-\left.\frac{\partial V_P(\omega_P,x) }{\partial x}\right|_{x=\omega_D}f_P(\omega_P).
\end{align*}

To illustrate how Theorem \ref{T:2} applies when $F$ has atoms, consider the example in \citename{KG} (\citeyear*{KG}, p.~2591) in which a prosecutor (principal) persuades a judge (agent) to convict a suspect. The suspect is innocent $(\omega=0)$ with probability 0.7 and guilty $(\omega=1)$ with probability 0.3, so that the distribution of $\omega$ is
\[
F(\omega)=\begin{cases}
0.7, & \text{if $\omega\in [0,1)$,}\\
1, & \text{if $\omega= 1$.}
\end{cases}
\]
The prosecutor's preferred decision is to convict the suspect irrespective of the state, whereas the judge's preferred decision is to convict the suspect whenever his posterior that the suspect is guilty is at least $1/2$,
\[
V_P(\omega,x)=x \quad\text{and}\quad U_P(\omega,x)=(\omega-1/2)x,
\]
where $\omega,x\in[0,1]$.
Substituting $\omega=F^{-1}(\theta)$ yields 
\[
V_P(\theta,x)=x \quad\text{and} \quad U_P(\theta,x)=\begin{cases}
-x/2, & \text{if $\theta\in[0,0.7)$},\\
x/2, & \text{if $\theta\in[0.7,1]$}.
\end{cases}
\]
Clearly, $(U_P,V_P)\in\bar{\mathcal P}_P$. By Theorem \ref{T:2}, an equivalent primitive $(U_D,V_D)\in\bar{\mathcal P}_D$ of the balanced delegation problem is given by
\[
V_D(\theta,x)=1-x \quad\text{and}\quad U_D(\theta,x)=\begin{cases}
(x-0.4)/2, & \text{if $x\in[0,0.7)$},\\
(1-x)/2, & \text{if $x\in[0.7,1]$}.
\end{cases}
\]
Assume that the agent breaks ties in favor of the principal. It is easy to see that the optimal balanced delegation set is $\Pi^*=\{0,0.4,1\}$, so that the agent is indifferent between $x=0.4$ and $x=1$ and thus chooses the principal's preferred decision, $x^*=0.4$. The principal's optimal expected payoff is $1-x^*=0.6$. 

Let us interpret the solution $\Pi^*=\{0,0.4,1\}$ within the persuasion problem. When the state $\theta$ belongs to the partition element $[0,0.4)$, the posterior is that $\omega=F^{-1}(\theta)=0$ with probability one, so the judge acquits the suspect. When the state $\theta$ belongs to the partition element $[0.4,1)$, the posterior is that $\omega=F^{-1}(\theta)=0$ with probability $1/2$, when $\theta\in[0.4,0.7)$, and $\omega=F^{-1}(\theta)=1$ with probability $1/2$, when $\theta\in[0.7,1)$. So, the judge is indifferent between convicting and acquiting, and thus convicts the suspect. The prosecutor's expected payoff is $0.6$, and $\Pi^*$ is the optimal information structure derived in \citename{KG} (\citeyear*{KG}, p.~2591).

\section{Concluding Remarks}\label{s:conc}

We have shown the equivalence of balanced delegation and monotone persuasion, with the upshot that insights in delegation can be used for better understanding of persuasion, and vice versa. For instance, persuasion as the design of a distribution of posterior beliefs is notoriously hard to explain to a non-specialized audience. The connection to delegation can thus be instrumental in relaying technical results from the persuasion literature to practitioners and policy makers.

We have used the tools from the literature on linear persuasion to obtain new results on linear delegation. \citeasnoun{AB} have developed a Lagrangian method to derive sufficient conditions for the optimality of interval delegation in a nonlinear delegation problem. This method may be useful for deriving conditions for the optimality of interval persuasion in an equivalent nonlinear persuasion problem.

The classical delegation and persuasion problems have numerous extensions, which include a privately informed principal, competing principals, multiple agents, repeated interactions, and multidimensional state and decision spaces. We hope that our equivalence result will be a starting point for studying the connection between delegation and persuasion in these extensions.   

It may be interesting to compare the values of delegation and persuasion in a given problem. This comparison can be made by recasting the persuasion problem as an equivalent delegation problem and then directly comparing the solutions and values of these two delegation problems.

Naturally, a principal may wish to influence an agent's decision by a combination of persuasion and delegation instruments. How to optimally control both information and decisions of the agent, how these instruments interact, and whether they are substitutes or complements are important questions that are left for future research.

\section*{Appendix}
\renewcommand{\thesection}{A}

\subsection{Interpretation of Persuasion Problem}\label{s:int}
In Section \ref{s:applic}, we have expressed the monopoly regulation problem as a balanced delegation problem and derived an equivalent monotone persuasion problem. The primitive of this problem, up to a multiplicative constant, is
\beq\label{e:mon-1}
U_P(\theta,x)=\frac{1}{2}\int_0^{x} (2\theta-1-F^{-1}(\gamma))\df \gamma \quad\text{and}\quad V_P(\theta,x)=\int_0^{x} (\theta-F^{-1}(\gamma))\df \gamma.
\eeq
We now provide an interpretation of this problem.

A producer (agent) chooses a quantity $x$ to produce. He faces uncertainty about an exogenous price $\theta$ that is uniformly distributed on $[0,1]$. %
The producer's payoff is
\beq\label{e:monU-1}
U_P(\theta,x)=\theta x-C(x),
\eeq
where $C(x)$ is the producer's cost function. A government agency (principal) can disclose information about the price to the producer. The agency's payoff is
\beq\label{e:monV-1}
V_D(\theta,x)=\theta x-D(x),
\eeq
where is $D(x)$ is the social cost function. The difference, $C(x)-D(x)$, is the producer's externality. The agency chooses $\Pi\in{\bf \Pi}$. The producer observes the partition element $\mu_\Pi(\theta)$ that contains the price $\theta$ and chooses a quantity $x$ that maximizes his expected payoff given the posterior belief about the price.

Observe that the payoffs \eqref{e:monU-1}--\eqref{e:monV-1} are the same as \eqref{e:mon-1} if 
\[
C(x)=\int_0^x\frac{1+F^{-1}(\gamma)}{2}\df \gamma \quad \text{and} \quad D(x)=\int_0^x F^{-1}(\gamma)\df \gamma.
\]
Both cost functions $C$ and $D$ are convex, since $F^{-1}$ is increasing. Notice that $C(x)-D(x)>0$ for all $x>0$. So, in this problem, we have a positive externality, where the social cost is smaller than the producer's cost.

\subsection{Proof of Propositions \ref{P:st} and \ref{P:st1}}

The derivative of $\nu(m)$ is
\[\label{Charles}
\nu'(m)=2(1-m)f(2m-1)+\int_0^{2m-1} f(\gamma)\df \gamma.
\]

First, we show that $\nu(m)$ is $S$-shaped (that is, $\nu'$ is single-peaked with an interior peak) when the density $f$ is unimodal.
\begin{lemma}\label{L:Unim}
Let $\gamma_m\in(0,1)$ be the mode of the density $f$. %
Then, $\nu(m)$ is convex on  $(-\infty,(1+\gamma_m)/2]$ (strictly so on $[1/2,(1+\gamma_m)/2]$) and concave on $[(1+\gamma_m)/2,\infty)$ (strictly so on $[(1+\gamma_m)/2,1]$).
\end{lemma}
\begin{proof}
For any $m_1,m_2$,
\bea
\nu'(m_2)-\nu'(m_1)&=&2(1-m_2)[f(2m_2-1)-f(2m_1-1)]\\
&&+\int_{2m_1-1}^{2m_2-1} [f(\gamma)-f(2m_1-1)]\df \gamma.
\eea
Thus, if $m_1<m_2\le (1+\gamma_m)/2$, then $\nu'(m_2)\geq\nu'(m_1)$, because $f(2m-1)$ is increasing for $m\in (-\infty,(1+\gamma_m)/2]$. Moreover, if $1/2\le m_1<m_2\le (1+\gamma_m)/2$, then $\nu'(m_2)>\nu'(m_1)$, because $f(2m-1)$ is strictly increasing for $m\in[1/2,(1+\gamma_m)/2]$. Similarly, if $(1+\gamma_m)/2\le m_1<m_2$, then $\nu'(m_2)\le \nu'(m_1)$, because $f(2m-1)$ is decreasing for $m\in[(1+\gamma_m)/2,\infty)$. Moreover, if $(1+\gamma_m)/2\le m_1<m_2\le 1$, then $\nu'(m_2)<\nu'(m_1)$, because $f(2m-1)$ is strictly decreasing for $m\in[(1+\gamma_m)/2,1]$.
\end{proof}

Propositions \ref{P:st} and \ref{P:st1} are special cases of Lemma \ref{P:st-2}.
\begin{lemma}\label{P:st-2}
Let $\gamma_m\in(0,1)$ be the mode of the density $f$, and let $\theta$ be uniformly distributed on $[0,\ol \theta]$ with $\ol \theta\ge 1$. The set $\Pi^\star=[0,\theta^\star]\cup\{\ol \theta\}$ is optimal, where $\theta^\star\in \big(\max\big\{0,1+\gamma_m-\ol \theta\big\},(1+\gamma_m)/2\big)$ is the unique solution to
\beq\label{Eq:UC-a}
\nu\left(\frac{\ol \theta+\theta^\star}{2}\right)- \nu\left(\theta^\star\right)=\left(\frac{\ol \theta+\theta^\star}{2}-\theta^\star\right)\nu'\left(\frac{\ol \theta+\theta^\star}{2}\right).
\eeq
\end{lemma}
\begin{proof}
It is straightforward to show (see Figures~\ref{F:A}(a) and \ref{F:A}(b)) that there exists a unique solution $\theta^\star$ to \eqref{Eq:UC-a} with 
\[
1-\E[\theta]\leq \theta^\star < \frac{1+\gamma_m}{2}<  \frac{\ol \theta+\theta^\star}{2}.
\] We now use Theorem \ref{T:3} in Section \ref{s:LDLP} to verify that $\Pi^\star$ is optimal. For $\Pi^\star=[0,\theta^\star]\cup\{\ol \theta\}$, we have
\[
p_{\Pi^\star}(m) = \begin{cases}
	\nu(m), & \text{if $m <\theta^\star$},\\
	\nu\left(\tfrac{\ol \theta+\theta^\star}{2}\right)+\left(m-\tfrac{\ol \theta+\theta^\star}{2}\right)\nu'\left(\tfrac{\ol \theta+\theta^\star}{2}\right), & \text{if $m\geq\theta^\star$}.
\end{cases}
\]
Figures~\ref{F:A}(a) and \ref{F:A}(b) show $p_{\Pi^\star}$ as a dashed red curve.
It is straightforward to verify that $p_{\Pi^\star}(m)$ is convex and $p_{\Pi^\star}(m)\geq \nu(m)$ for all $m\in[0,\ol \theta]$. Thus, $\Pi^\star$ is optimal by Theorem~\ref{T:3}.
\end{proof}

\subsection{Standard Delegation}\label{s:StdDel}

Consider a delegation problem in which the set of states is $\Theta=[\ul\theta,\ol\theta]$ and the set of decisions is the real line. The principal chooses $\Pi\in {\bf\Pi}(\R)$, where ${\bf\Pi}(\R)$ is the set of all compact subsets of $\R$. Payoffs $U_D$ and $V_D$ satisfy assumptions $(\bar{\rm A}_1)$ and $(\bar{\rm A}_2^D)$. In addition, we assume that

($\bar{\rm A}_3$) $\sup\limits_{\theta\in\Theta}U_D(\theta,x)\to-\infty$  and $\sup\limits_{\theta\in\Theta}V_D(\theta,x)\to-\infty$ as $x\to\pm\infty$.

Note that $(\bar{\rm A}_1)$, $(\bar{\rm A}_2^D)$, and $(\bar{\rm A}_3)$ are satisfied in the linear delegation problem in Section~\ref{s:LDLP}.

We now show that this problem can be formulated as a balanced delegation problem, up to rescaling of the decision, in which the principal chooses $\Pi \in {\bf \Pi}([\ul y,\ol y])$ given by \eqref{E:Pi} for a sufficiently large compact set $[\ul y, \ol y]$.

\begin{lemma}\label{L:3}
There exists an interval $[\ul x,\ol x]$ such that, for each $\ul y< \ul x$ and each $\ol y> \ol x$, 
\[
\max_{\Pi\in {\bf \Pi}([\ul y,\ol y])}\E \big[V_D(\theta,x^*_D(\theta,\Pi))\big]=\max_{\Pi\in{\bf \Pi}(\R)}\E \big[V_D(\theta,x^*_D(\theta,\Pi))\big].
\]
\end{lemma}
\begin{proof}
Consider $z_0\in \R$. Let $V_0$ be the principal's expected payoff if $\Pi=\{z_0\}$,
\[
V_0=\E[V_D(\theta,z_0)].
\]
Let 
\[
Z=\left\{x\in\R: \sup_{\theta\in\Theta} V_D(\theta, x)\ge V_0\right\}.
\]
Note that $Z$ is nonempty because  $z_0\in Z$, and it is bounded by $(\bar{\rm A}_3)$. Let $\ul z =\inf Z$ and $\ol z = \sup Z$.
Clearly, for each $\Pi\in{\bf \Pi}(\R)$,
\[
\Pi\cap [\ul z,\ol z] =\varnothing \implies \E\big[V_D(\theta,x^*_D(\theta,\Pi))\big]<V_0,
\]
so an optimal $\Pi^*\in{\bf \Pi}(\R)$ must have a nonempty intersection with $[\ul z,\ol z]$.

Next, we say that $y\in \R$ {\it is dominated by} $[\ul z,\ol z]$ if the agent strictly prefers any decision in $[\ul z,\ol z]$ to $y$,
\[
\min_{z\in [\ul z,\ol z]} U_D(\theta,z)-U_D(\theta,y)>0 \ \ \text{for all $\theta\in\Theta$}.
\]
By $(\bar{\rm A}_1)$ and $(\bar{\rm A}_2^D$, downcrossing in $x$),
\beq\label{E:MIT0}
\min_{z\in[\ul z,\ol z]} U_D(\theta,z)=\min\{U_D(\theta,\ul z),U_D(\theta,\ol z)\}.
\eeq
By $(\bar{\rm A}_1)$ and $(\bar{\rm A}_2^D$, aggregate upcrossing in $\theta$),
\[
U_D(\theta,z)-U_D(\theta,y)=\int_{y}^{z} \frac{\partial}{\partial x}U_D(\theta,x)\df x>0 \ \ \text{for all $\theta\in[\ul\theta,\ol \theta]$}
\]
if and only if
\beq\label{E:MIT1}
\begin{cases}
U_D(\ul \theta,z)-U_D(\ul \theta,y)>0 &  \ \ \text{if $z>y$},\\
U_D(\ol \theta,z)-U_D(\ol \theta,y)>0 & \ \  \text{if $z<y$}.
\end{cases}
\eeq

Let $X$ be the set of all $x\in \R$ that are not dominated by $[\ul z,\ol z]$. If $x\in [\ul z,\ol z]$, then, trivially, $x\in X$. If $x<\ul z$, then, by \eqref{E:MIT0} and \eqref{E:MIT1},
\[
x\in X \iff U_D(\ul \theta,x)\ge \min\{U_D(\ul \theta,\ul z),U_D(\ul \theta,\ol z)\}.
\]
If $x>\ol z$, then, by \eqref{E:MIT0} and \eqref{E:MIT1},
\[
x\in X \iff U_D(\ol \theta,x)\ge \min\{U_D(\ol \theta,\ul z),U_D(\ol \theta,\ol z)\}.
\]
Thus, $[\ul z,\ol z]\subset X$ and, by $(\bar{\rm A}_3)$, $X$ is bounded. 

So, we have obtained that (i) if $\Pi^*$ is optimal, then it has a nonempty intersection with $[\ul z,\ol z]$, and (ii) any decision $x\not\in X$ is dominated by $[\ul z,\ol z]$. Given that $\Pi^*\cap [\ul z,\ol z]\neq\varnothing$, adding or removing any decisions outside of $X$ does not affect the agent's behavior, 
\[
x^*_D(\theta,\Pi^*)=x^*_D(\theta,\Pi^*\cup Y)=x^*_D(\theta,\Pi^*\backslash Y) \ \ \text{for all $\theta$ and $Y\subset \R\backslash X$}.
\]
Hence, for each $\ul y<\min X$ and each $\ol y>\max X$,
\[
\E\big[V_D(\theta,x^*_D(\theta,\Pi^*))\big] = \E\big[V_D(\theta,x^*_D(\theta,\Pi^*\cap [\ul y,\ol y]\cup\{\ul y,\ol y\}))\big].\qedhere
\]
\end{proof}

\subsection{Proof of Proposition \ref{P:D}.}\hfill

$1.$ For $\Pi^*=\{x^*,\ul y, \ol y\}$, we have
\[
p_{\Pi^*}(m) = 
\begin{cases}
\nu(m^*_L)+\nu'(m^*_L) (m - m^*_L), &\text{if $m< c(x^*)$},\\
\nu(m^*_H) + \nu'(m^*_H) (m - m^*_H), &\text{if $m\geq c(x^*)$},
\end{cases}
\]
where
\begin{align*}
	m^*_L = \E [c(s)|s< x^*] \quad\text{and}\quad	m^*_H &= \E [c(s)|s\ge x^*].
\end{align*}
As shown in Appendix \ref{s:StdDel}, $\ul y$ is sufficiently small and $\ol y$ is sufficiently large, such that
\begin{align*}
	m^*_L &=\frac{\int_{\ul y}^{x^*}c(s)\df s}{x^*- \ul y}  = \frac{\int_{\ul y}^{x'}c(s)\df s + \int_{x'}^{x^*}c(s)\df s}{x'-\ul y + x^* - x'}<0,\\
	m^*_H &=\frac{\int_{x^*}^{\ol y}c(s)\df s}{\ol y - x^*} = \frac{\int_{x^*}^{x''}c(s)\df s + \int_{x''}^{\ol y}c(s)\df s}{x'' - x^* + \ol y - x''}>1,
\end{align*}
where the first inequality holds because $c(x)<0$ for $x< x'$ and $\ul y\ll x'$ and the second inequality holds because $c(x)>1$ for $x> x''$ and $\ol y\gg x''$. 

Thus, taking into account \eqref{Eq:VM-6}, we have
\[
p_{\Pi^*}(m) = 
\begin{cases}
0, &\text{if $m< c(x^*)$},\\
\nu(1) + m - 1, &\text{if $m\geq c(x^*)$}. 
\end{cases}
\]
Figure~\ref {f:3}(a) shows $p_{\Pi^*}$ as a dashed red curve. The conditions imply that $p_{\Pi^*}$ is convex and $p_{\Pi^*}\geq \nu$; so $\Pi^*$ is optimal by Theorem~\ref{T:3}.

$2.$ For $\Pi^*=[x^*_L,x^*_H]\cup \{\ul y,\ol y\}$, we have
\[
p_{\Pi^*}(m) = 
\begin{cases}
\nu(m^*_L)+\nu'(m^*_L) (m - m^*_L), &\text{if $m< c(x^*_L)$},\\
\nu(m), &\text{if $m\in [c(x^*_L),c(x^*_H))$},\\
\nu(m^*_H) + \nu'(m^*_H) (m - m^*_H), &\text{if $m\geq c(x^*_H)$},
\end{cases}
\]
where
\begin{align*}
	m^*_L = \E [c(s)|s< x^*_L] \quad\text{and}\quad m^*_H = \E [c(s)|s\ge x^*_H].
\end{align*}
As shown in Appendix \ref{s:StdDel}, $\ul y$ is sufficiently small and $\ol y$ is sufficiently large, such that
\begin{align*}
m^*_L &=\frac{\int_{\ul y}^{x^*_L}c(s)\df s}{x^*_L- \ul y} = \frac{\int_{\ul y}^{x'}c(s)\df s + \int_{x'}^{x^*_L}c(s)\df s}{x'-\ul y + x^*_L - x'}<0,\\
	m^*_H &=\frac{\int_{x^*_H}^{\ol y}c(s)\df s}{\ol y - x^*_H} = \frac{\int_{x^*_H}^{x''}c(s)\df s + \int_{x''}^{\ol y}c(s)\df s}{x'' - x^*_H + \ol y - x''}>1,
\end{align*}
where the first inequality holds because $c(x)<0$ for $x< x'$ and $\ul y\ll x'$ and the second inequality holds because $c(x)>1$ for $x> x''$ and $\ol y\gg x''$. 

Thus, taking into account \eqref{Eq:VM-6}, we have
\[
p_{\Pi^*}(m) = 
\begin{cases}
0, &\text{if $m< c(x^*_L)$},\\
\nu(m), &\text{if $m\in [c(x^*_L),c(x^*_H))$},\\
\nu(1) + m - 1, &\text{if $m\geq c(x^*_L)$}. 
\end{cases}
\]
Figure~\ref{f:3}(b) shows $p_{\Pi^*}$ as a dashed red curve. The conditions imply that $p_{\Pi^*}$ is convex and $p_{\Pi^*}\geq \nu$; so $\Pi^*$ is optimal by Theorem~\ref{T:3}. In particular, if $c(x_L^*)\in (0,1)$, then $p_{\Pi^*}(m)$ is convex at $m=c(x_L^*)$ because $\nu(m)$ is differentiable at all $m\in (0,1)$, by \eqref{Eq:VM-6}. Moreover, if $c(x_L^*)= 0$, then $p_{\Pi^*}(m)$ is convex at $m=0$ because $\nu(m)$ is convex at $m=0$, by the last line in the conditions. By the same argument, $p_{\Pi^*}(m)$ is convex at $m=c(x_H^*)$.

$3.$ For $\Pi^*=[\ul y,x^*_L]\cup [x^*_H, \ol y]$, we have
\[
p_{\Pi^*}(m) = 
\begin{cases}
\nu(m), &\text{if $m< c(x^*_L)$},\\
\nu(m^*)+\nu'(m^*)(m-m^*), &\text{if $m\in [c(x^*_L),c(x^*_H))$},\\
\nu(m), &\text{if $m\geq c(x^*_H)$}.
\end{cases}
\]
Figure~\ref{f:3}(c) shows $p_{\Pi^*}$ as a dashed red curve. The conditions imply that $p_{\Pi^*}$ is convex and $p_{\Pi^*}\geq \nu$; so $\Pi^*$ is optimal by Theorem~\ref{T:3}. As in part~2, the last line in the conditions verifies that $p_{\Pi^*}(m)$ is convex at $m=c(x^*_L)$ and at $m=c(x^*_H)$ if $\nu(m)$ is not differentiable at these points.

\subsection{Proof of Proposition \ref{P:Infl}.}
The proof is straightforward but tedious. We only summarize possible cases. The reader may refer to the corresponding figures for guidance.

$1.$ There are 4 cases (see Figure~\ref{f:4}).

$(a)$ If $\nu'(0+)\geq 0$ and $\nu'(1-)\leq 1$, then $\Pi^*=[\ul y, \ol y]$. 

$(b)$ If $\nu'(0+)\geq 0$ and $\nu'(1-)> 1$, then $\Pi^*=[\ul y, x^*_H]\cup \{\ol y\}$ with $c(x^*_H)\leq 1$. 

$(c)$ If $\nu'(0+) < 0$ and $\nu'(1-)\leq 1$, then $\Pi^*=[x^*_L,\ol y]\cup \{\ul y\}$ with $c(x^*_L)\geq 0$.

$(d)$ If $\nu'(0+) < 0$ and $\nu'(1-) > 1$, then $\Pi^*=[x^*_L,x^*_H]\cup \{\ul y, \ol y\}$ with $0\leq c(x^*_L)\leq c(x^*_H)\leq 1$.

$2.$ There are 3 cases (see Figure~\ref{f:5}).

$(a)$ If $\nu'(0+)> 0$ and $\nu'(1-)< 1$, then $\Pi^*=\{x^*_L,x^*_H,\ul y, \ol y\}$ with $c(x^*_L)\leq 0$ and $c(x^*_H) \geq 1$. 

$(b)$ If $\nu'(0+)\leq 0$, then $\Pi^*=\{x^*_H,\ul y, \ol y\}$ with $c(x^*_H)\geq 1$. 

$(c)$ If $\nu'(1-)\geq 1$, then $\Pi^*=\{x^*_L, \ul y, \ol y\}$ with $c(x^*_L) \leq 0$. 

\begin{figure}
\begin{tabular}{cc}
\includegraphics[width=170pt]{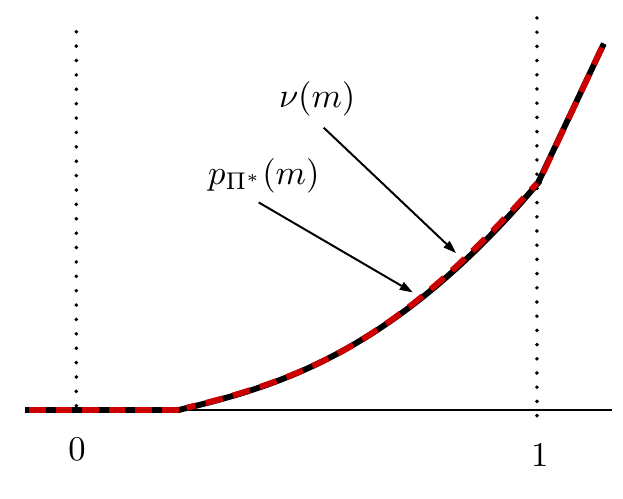} & \includegraphics[width=170pt]{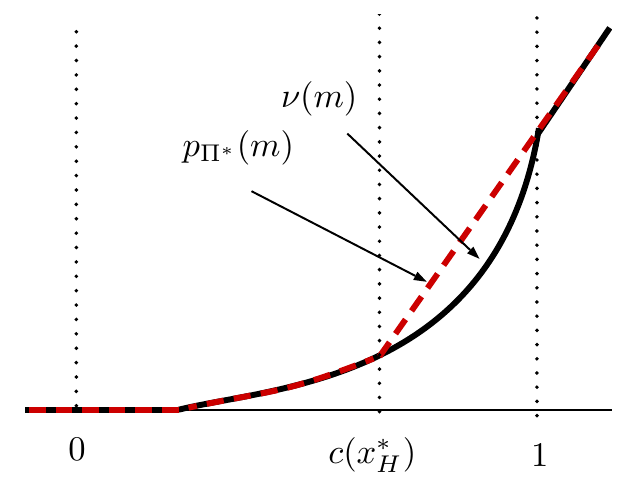}\\
\footnotesize{(a)}&\footnotesize{(b)}\\ \ &\\
\includegraphics[width=170pt]{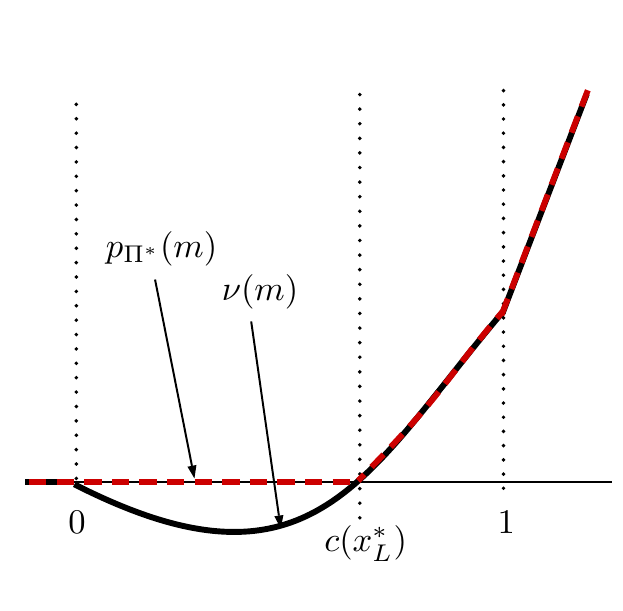} & \includegraphics[width=170pt]{fm2.pdf}\\
\footnotesize{(c)}&\footnotesize{(d)}\\ \ & \\
\end{tabular}
\caption{Four Cases in Proposition \ref{P:Infl} (Part 1)}
\label{f:4}
\end{figure}

\begin{figure}
\begin{tabular}{ccc}
\hspace*{-0.3in}\includegraphics[width=150pt]{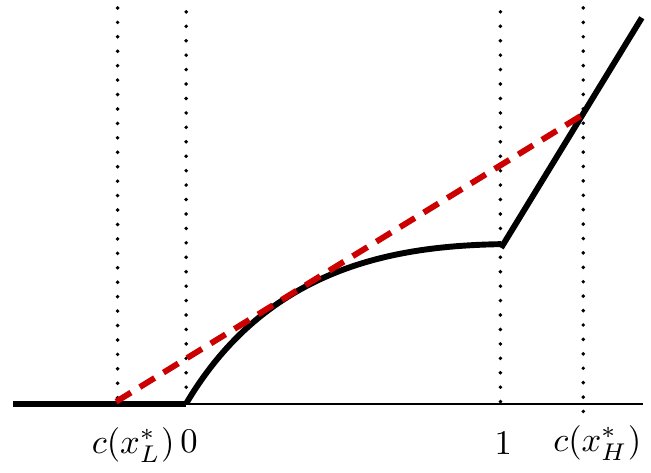} & \hspace*{-0.1in}\includegraphics[width=150pt]{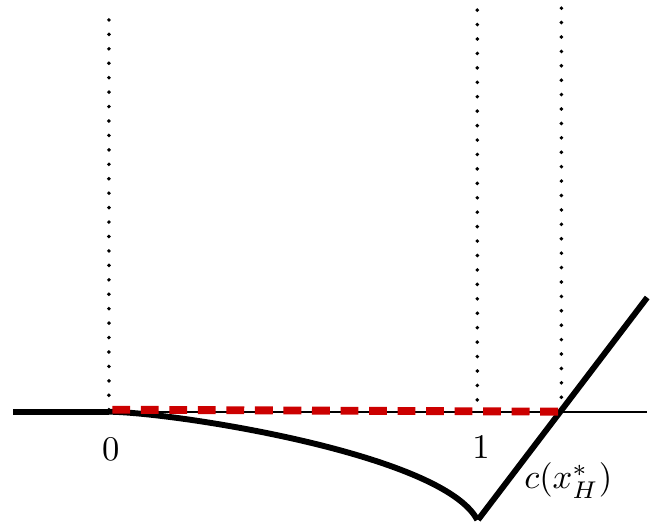} & \hspace*{-0.1in}\includegraphics[width=150pt]{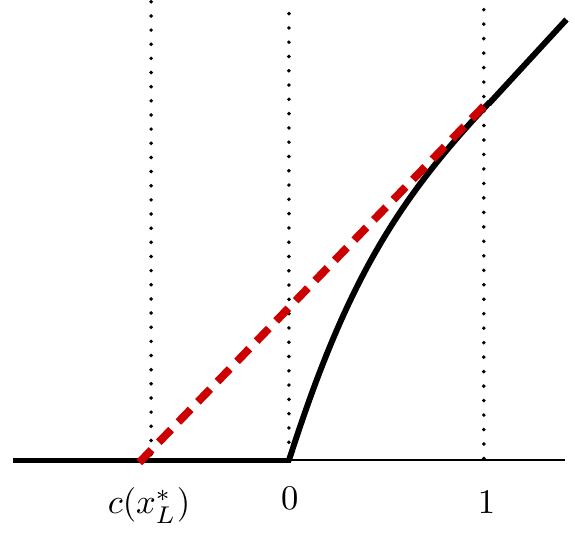}\\
\footnotesize{(a)}&\footnotesize{(b)} &\footnotesize{(c)}
\end{tabular}
\caption{Three Cases in Proposition \ref{P:Infl} (Part 2)}
\label{f:5}
\end{figure}

\begin{figure}
\begin{tabular}{cc}
\includegraphics[width=200pt]{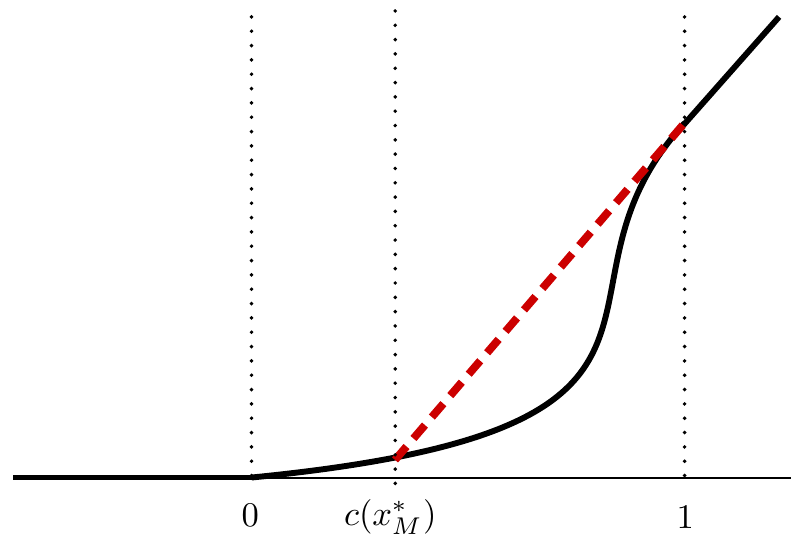} & \includegraphics[width=200pt]{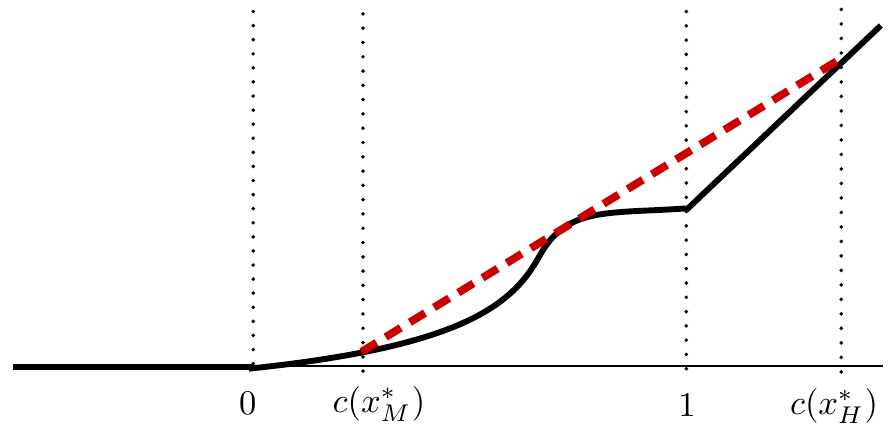}\\
\footnotesize{(a)}&\footnotesize{(b)}\\ \ &\\
\includegraphics[width=200pt]{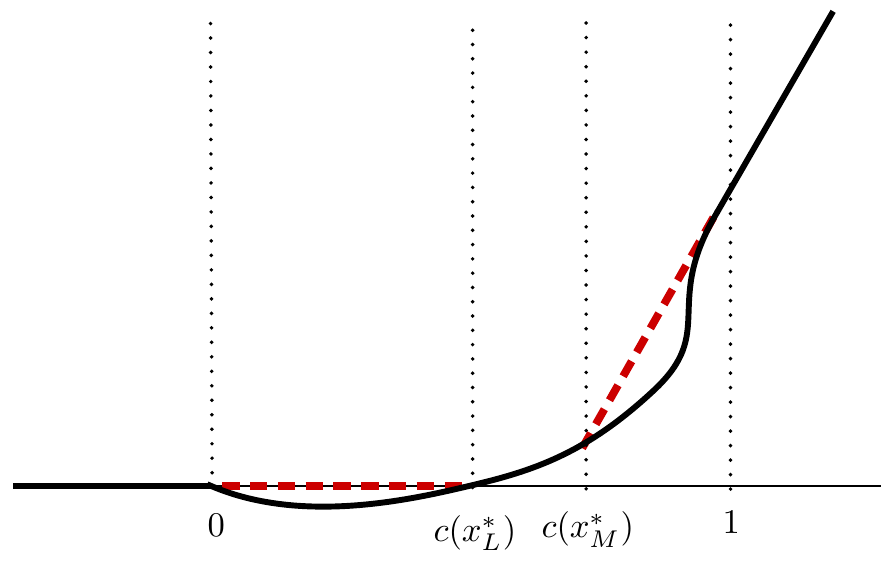} & \includegraphics[width=200pt]{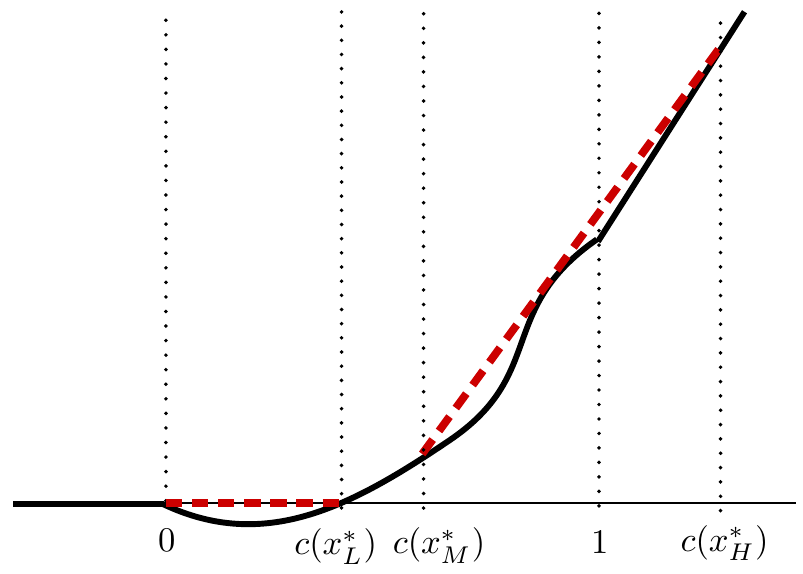}\\
\footnotesize{(c)}&\footnotesize{(d)}
\end{tabular}
\caption{Four Cases in Proposition \ref{P:Infl} (Part 3)}
\label{f:6}
\end{figure}

$3.$ There are 4 cases (see Figure~\ref{f:6}).

$(a)$ If $\nu'(0+)\geq 0$ and $\nu'(1-) \geq 1$, then $\Pi^*=[\ul y, x^*_M]\cup \{\ol y\}$ with $c(x^*_M)\leq\tilde m$. 

$(b)$ If $\nu'(0+)\geq 0$ and $\nu'(1-) < 1$, then $\Pi^*=[\ul y, x^*_M]\cup \{x^*_H,\ol y\}$ with $c(x^*_M)\leq \tilde m$ and $c(x^*_H)\geq 1 $.

$(c)$ If $\nu'(0+)< 0$ and $\nu'(1-) \geq 1$, then $\Pi^*=[x^*_L,x^*_M]\cup \{\ul y, \ol y \}$ with $0\leq c(x^*_L)\leq c(x^*_M)\leq \tilde m$.

$(d)$ If $\nu'(0+)< 0$ and $\nu'(1-) < 1$, then $\Pi^*=[x^*_L, x^*_M]\cup \{x^*_H,\ul y,\ol y\}$ with $0\leq c(x^*_L)\leq c(x^*_M)\leq\tilde m$ and $c(x^*_H)\geq 1$.

$4.$ There are 4 cases analogous to those in part 3.

$(a)$ If $\nu'(0+)\leq 0$ and $\nu'(1-) \geq 1$, then $\Pi^*=[x^*_M, \ol y]\cup \{\ul y\}$ with $c(x^*_M)\geq\tilde m$.

$(b)$ If $\nu'(0+)> 0$ and $\nu'(1-) \geq 1$, then $\Pi^*=[x^*_M, \ol y]\cup \{x^*_L,\ul y\}$ with $c(x^*_M)\geq\tilde m$ and $c(x^*_L)\leq 0$.

$(c)$ If $\nu'(0+)\leq 0$ and $\nu'(1-) > 1$, then $\Pi^*=[x^*_M,x^*_H]\cup \{\ul y, \ol y \}$ with $1\geq c(x^*_H)\geq c(x^*_M)\geq \tilde m$.

$(d)$ If $\nu'(0+)> 0$ and $\nu'(1-) < 1$, then $\Pi^*=[x^*_M, x^*_H]\cup \{x^*_L,\ul y,\ol y\}$ with $1\geq c(x^*_H)\geq c(x^*_M)\geq\tilde m$ and $c(x^*_L)\leq 0$.

\subsection{Proof of Proposition \ref{P:7}}
For $\Pi^*=\{x_0 \}\cup [x^*, \ol y]$, we have
\[
p_{\Pi^*}(m) = 
\begin{cases}
\nu(m^*)+\nu'(m^*)(m-m^*), &\text{if $m\in [c(x_0),c(x^*))$},\\
\nu(m), &\text{if $m\geq c(x^*)$}.
\end{cases}
\]
The conditions imply that $p_{\Pi^*}$ is convex and $p_{\Pi^*}\geq \nu$; so $\Pi^*$ is optimal by Theorem~\ref{T:3}. As in Proposition~\ref{P:D}, the last line in the conditions verifies that $p_{\Pi^*}(m)$ is convex at $m=c(x^*)$ if $\nu(m)$ is not differentiable at this point.

\subsection{Proof of Theorem \ref{T:2}}
Consider $(U_D,V_D)\in\bar{\mathcal P}_D$ and $(U_P,V_P)\in\bar{\mathcal P}_P$ that satisfy 
\begin{align}
\left.\frac{\partial U_D(t,x)}{\partial x}\right|_{x=s} = - \left.\frac{\partial U_P(s,x)}{\partial x}\right|_{x=t} \ \ &\text{and} \ \ \left.\frac{\partial V_D(t,x)}{\partial x}\right|_{x=s}=-\left.\frac{\partial V_P(s,x)}{\partial x}\right|_{x=t}\label{E:Og1}
\end{align}
for all $s,t\in[0,1]$. 
It suffices to prove that there exists a constant $\beta \in \R$ such that 
\[
\E\big[V_D(\theta,x^*_D(\theta,\Pi))\big]=\E\big[V_P(\theta,x^*_P(\theta,\Pi))\big]+\beta\text{ for all $\Pi\in\bf \Pi$}.
\]

Consider $\Pi\in \bf\Pi$ and let $s$ be uniformly distributed on $[0,1]$. Define
\begin{align*}
&u_\Pi(s,t)=\E\left[\left.\frac{\partial U_P(s',t)}{\partial t}\right|s'\in\mu_\Pi(s)\right]=\E\left[\left.-\frac{\partial U_D(t,s')}{\partial s'}\right|s'\in\mu_\Pi(s)\right],\\
&v_\Pi(s,t)=\E\left[\left.\frac{\partial V_P(s',t)}{\partial t}\right|s'\in\mu_\Pi(s)\right]=\E\left[\left.-\frac{\partial V_D(t,s')}{\partial s'}\right|s'\in\mu_\Pi(s)\right].
\end{align*}
Note that $u_\Pi(s,t)$ is integrable in $s$ and $t$ by ($\bar{\rm A}_1$) and satisfies upcrossing in $s$ and downcrossing in $t$ by ($\bar{\rm A}_2^D$)/($\bar{\rm A}_2^P$).

First, consider the balanced delegation problem. By  \eqref{E:Og1}, we have, for $s\in \Pi$,
\begin{align*}
U_D(t,s)&=U_D(t,1)-\int_{s}^1 \left.\frac{\partial U_D(t,x)}{\partial x}\right|_{x=s'}\df s'=U_D(t,1)+\int_{s}^1 u_\Pi(s',t)\df s',\\
V_D(t,s)&=V_D(t,1)-\int_{s}^1 \left.\frac{\partial V_D(t,x)}{\partial x}\right|_{x=s'}\df s'=V_D(t,1)+\int_{s}^1 v_\Pi(s',t)\df s'.
\end{align*}

Since $u_\Pi(s,t)$ satisfies upcrossing in $s$, we have 
\[
s\in y^*(t,\Pi)=\argmax_{y\in \Pi}\int_{y}^1 u_\Pi(s,t)\df s
\]
if and only if for all $s',s''\in \Pi$  such that $s'\le s\leq s''$ we have $u_\Pi(s',t)\leq 0 \leq u_\Pi(s'',t)$.

The principal's expected payoff is
\begin{align*}
\E\big[V_D(t,x^*_D(t,\Pi))\big]&=\E[V_D(t,1)]+\E[V_D(t,y^*(t,\Pi))]\\
&=\E[V_D(t,1)]+\int_0^1\int_{y^*(t,\Pi)}^1 v_\Pi(s,t)\df s\df t.
\end{align*}

Define
\begin{gather*}
J_\Pi=\{(s,t)\in[0,1]^2: u_\Pi(s,t)\ge 0\} \quad\text{and}\quad J_\Pi^+=\{(s,t)\in[0,1]^2: u_\Pi(s,t)>0\}, \\
\mathcal J_\Pi=\{J\subset[0,1]^2:  J_\Pi^+\subset J\subset J_\Pi, \ \text{and $J$ is Lebesgue measurable}\}.
\end{gather*}

Using \citename{Aumann} (\citeyear*{Aumann}, Theorem 1), we have
\[
\E\big[V_D(t,x^*_D(t,\Pi))\big]=\E[V_D(t,1)]+\left\{\int_{(s,t)\in J} v_\Pi(s,t)\df s\df t: J\in\mathcal J_\Pi\right\}.
\]

Next, consider the monotone persuasion problem. By  \eqref{E:Og1}, we have
\begin{align*}
\E\big [U_P(s',t) \big| s'\in \mu_\Pi(s)\big]&=\E\big [U_P(s',0) \big| s'\in \mu_\Pi(s)\big]+
\int_0^t\E\left[\left.\frac{\partial U_P(s',t)}{\partial t} \right| s'\in \mu_\Pi(s)\right]\df t'\\
&=\E\big [U_P(s',0) \big| s'\in \mu_\Pi(s)\big]+\int_0^t u_\Pi(s,t')\df t'
\end{align*}
and
\begin{align*}
\E\big [V_P(s',t) \big| s'\in \mu_\Pi(s)\big]&=\E\big [V_P(s',0) \big| s'\in \mu_\Pi(s)\big]+
\int_0^t\E\left[\left.\frac{\partial V_P(s',t)}{\partial t} \right| s'\in \mu_\Pi(s)\right]\df t'\\
&=\E\big [V_P(s',0) \big| s'\in \mu_\Pi(s)\big]+\int_0^t v_\Pi(s,t')\df t'.
\end{align*}

Since $u_\Pi(s,t)$ satisfies downcrossing in $t$, we have
\[
t\in z^*(s,\Pi)=\argmax_{z\in[0,1]}\int_0^z u_\Pi(s,t)\df t
\]
if and only if for all $t',t''\in [0,1]$  such that $t'\le t\leq t''$ we have $u_\Pi(s,t')\geq 0 \geq u_\Pi(s,t'')$. 

The principal's expected payoff is
\begin{align*}
\E[V_P(s,x^*_P(s,\Pi))]&=\E\big [V_P(s,0) \big]+\E[V_P(s,z^*(s,\Pi))]\\
&=\E[V_P(s,0)]+\int_{0}^1\int_0^{z^*(s,\Pi)} v_\Pi(s,t')\df t'\df s\\
&=\E[V_P(s,0)]+\left\{\int_{(s,t)\in J} v_\Pi(s,t)\df s\df t: J\in\mathcal J_\Pi\right\}.
\end{align*}
We thus have shown that $\E\big[V_D(\theta,x^*_D(\theta,\Pi))\big]=\E[V_P(\theta,x^*_P(\theta,\Pi))]+\beta$ for each $\Pi\in\bf\Pi$, where $\beta=\E[V_D(t,1)]-\E[V_P(s,0)]$.

Finally, it is straightforward to verify that if $(U_D,V_D)\in\bar{\mathcal P}_D$ and $(U_P,V_P)$ is given by \eqref{E:DtoP}, then $(U_P,V_P)$ satisfies ($\bar{\rm A}_1$) and ($\bar{\rm A}_2^P$), and thus it is in $\bar{\mathcal P}_P$. Conversely, if $(U_P,V_P)\in\bar{\mathcal P}_P$ and $(U_D,V_D)$ is given by \eqref{E:PtoD}, then $(U_D,V_D)$ satisfies ($\bar{\rm A}_1$) and ($\bar{\rm A}_2^D$), and thus it is in $\bar{\mathcal P}_D$.

\bibliographystyle{econometrica}
\bibliography{persuasionlit}

\end{document}